\documentclass[a4paper,UKenglish,cleveref, autoref, thm-restate]{lipics-v2021}

\newif\iflong
\newif\ifshort
\longtrue
  
\iflong
\else
\shorttrue
\fi

\usepackage{complexity}
\newcommand{\bigoh}{\mathcal{O}}

\newtheorem{longtheorem}{Theorem}

\newtheorem{longdefinition}[longtheorem]{Definition}
\newtheorem{longobservation}[longtheorem]{Observation}

\title{Bilateral Treewidth for QBF:\\  Where Strategies and Resolution Meet
} 

\titlerunning{Bilateral Treewidth for QBF: Where Strategies and Resolution Meet} 

\author{Robert Ganian}{TU Wien, Austria}{rganian@ac.tuwien.ac.at}{https://orcid.org/0000-0002-7762-8045}{} 
\author{Marlene Gr\"{u}ndel}{TU Wien, Austria}{mgruendel@ac.tuwien.ac.at}{https://orcid.org/0000-0003-3470-1326}{} 

\authorrunning{Robert Ganian and Marlene Gr\"{u}ndel} 

\Copyright{Robert Ganian and Marlene Gr\"{u}ndel}



\relatedversion{A short version of this article will appear in the proceedings of SAT 2026.}

\category{}  

 \hideLIPIcs
  
\funding{This research was funded in whole or in part by the Austrian Science Fund (FWF) 10.55776/COE12 and 10.55776/Y1329.} 

\nolinenumbers

 \usepackage{graphicx}  \usepackage{bussproofs}
\usepackage{amsmath}
\usepackage{ amssymb }
\usepackage{amsthm}
\usepackage{todonotes}
\usepackage{ wasysym }
\usepackage{enumitem}
\usepackage{algorithm}
\usepackage[noend]{algpseudocode}
\usepackage{tikz}
\usetikzlibrary{shapes}
\usetikzlibrary{arrows,arrows.meta,automata, positioning,calc, patterns,decorations.markings}
\usepackage{caption}
\usepackage{subcaption}
  
  \newcommand{\vara}[1]{\ensuremath{\mathrm{Var}_{\forall}({#1})}}
\newcommand{\vare}[1]{\ensuremath{\mathrm{Var}_{\exists}({#1})}}
\newcommand{\Q}{\ensuremath{\mathcal{Q}}}
\newcommand{\Z}{\ensuremath{\mathcal{Z}}}
\renewcommand{\A}{\ensuremath{\mathcal{A}}}
\newcommand{\B}{\ensuremath{\mathcal{B}}}
\renewcommand{\E}{\ensuremath{\mathcal{E}}}

\newcommand{\var}[1]{\ensuremath{\mathrm{Var}({#1})}}

\renewcommand{\D}{\ensuremath{\mathcal{D}}}
\newcommand{\Dtrv}{\ensuremath{\mathcal{D}_{\mathrm{trv}}}}
\newcommand{\V}{\ensuremath{\mathcal{V}}}
\newcommand{\clause}{\ensuremath{\mathcal{C}}}
\newcommand{\dep}{\ensuremath{\preceq_{\D}}}
\newcommand{\depless}[1]{\ensuremath{\mathcal{D}_{\preceq}({#1})}}

\newcommand{\qbf}{\ensuremath{\Phi=\mathcal{Q}.\phi}}
\newcommand{\red}[2]{\ensuremath{\mathrm{Red}({#1},{#2})}}
\newcommand{\expa}[2]{\ensuremath{\mathrm{Exp}({#1},{#2})}}
\newcommand{\res}[2]{\ensuremath{\mathrm{Res}({#1},{#2})}}
\newcommand{\strext}[2]{\ensuremath{\mathrm{Ext}({#1},{#2})}}
\renewcommand{\R}[2]{\ensuremath{\mathrm{R}({#1},{#2})}}

\newcommand{\pspace}{\ensuremath{\mathsf{PSPACE}}}
\newcommand{\qpar}{\ensuremath{\textsc{QParity}_n}}
\newcommand{\qsat}{\ensuremath{\textsc{QSat}}}
\newcommand{\sat}{\ensuremath{\textsc{Sat}}}
\newcommand{\qdres}{\ensuremath{\mathsf{Q}(\D)\text{-}\mathsf{Res}}}

 \newcommand{\lesst}{\ensuremath{<_{\mathbf{T}}}}
\newcommand{\lesstp}{\ensuremath{<_{\mathbf{T}'}}}
\newcommand{\lessp}{\ensuremath{<_{\mathbf{P}}}}
\newcommand{\tree}{\ensuremath{\mathcal{T}}}
\newcommand{\rt}{\ensuremath{\mathrm{root}(\mathbf{T})}}
\newcommand{\forget}[1]{\mathrm{forget}({#1})}

\DeclareRobustCommand{\rchi}{{\mathpalette\irchi\relax}}
\newcommand{\irchi}[2]{\raisebox{\depth}{$#1\chi$}}

\newcommand{\trunk}{\ensuremath{\mathrm{Tr}_{\lesst}}}
\newcommand{\btw}[1]{\ensuremath{\mathrm{btw}({#1})}}
\newcommand{\dtw}[1]{\ensuremath{\mathrm{dtw}({#1})}}
\newcommand{\ppw}[1]{\ensuremath{\mathrm{ppw}({#1})}}

\newcommand{\dpw}[1]{\ensuremath{\mathrm{dpw}({#1})}}
\newcommand{\cn}[1]{\ensuremath{\mathrm{cn}({#1})}}

 \newcommand{\oh}[1]{\ensuremath{\mathcal{O}({#1})}}

\makeatletter
\def\namedlabel#1#2{\begingroup
    #2     \def\@currentlabel{#2}     \phantomsection\label{#1}\endgroup
}

\begin{document}

\maketitle

 \begin{abstract}
Treewidth is a well-studied decompositional parameter to measure the tree-likeness of a graph. While the propositional satisfiability problem ($\sat$) is known to be tractable when parameterized by the treewidth of the underlying primal graph, the evaluation of quantified Boolean formulas (QBFs) remains \pspace-complete even on formulas of constant treewidth. 
Intuitively, this is because ordinary treewidth does not take into account the prefix of the QBF: it neither distinguishes between existential and universal variables, nor accounts for the order in which they are quantified.
In the past, several weaker variants of treewidth have been devised to incorporate prefix-sensitive information. 
To establish tractability for QBFs under these notions, prior work has employed either strategy- or resolution-based techniques, thereby dividing the parameterized complexity landscape of QBF into two regimes that are incomparable in strength.
We establish fixed-parameter tractability with respect to bilateral treewidth, a novel and strictly more powerful decompositional parameter that combines these rivaling approaches by simultaneously allowing for branching on strategies and performing Q-resolution.
As in previous works in this direction, our algorithm assumes that a suitable tree decomposition is provided on the input.
 \end{abstract}

\section{Introduction}
The language of \emph{quantified Boolean formulas} (QBFs)~\cite{BJLS21} extends propositional logic with existential and universal quantification over variables. While this enables succincter formalizations, which is highly desirable for tasks in planning and verification~\cite{BM08,BKS14,Z14,EKLP17}---see also~\cite{SBPS19} for a recent survey on further applications---it comes at the price of increased complexity: the satisfiability problem for QBFs ($\qsat$) is $\pspace$-complete~\cite{SM73}.

When tasked with solving $\qsat$, existential variables can often be handled analogously to the propositional setting ($\sat$), whereas genuinely new approaches are needed for their universal counterparts.
Over the past decades, two conceptually different paradigms have emerged, operating either \emph{top-down} or \emph{bottom-up}, respectively. In the top-down approach, \emph{universal expansion} explicitly branches over all (partial) assignments to universal variables, creating a (possibly exponentially large) conjunction of formulas in which existential variables are replaced by annotated copies to indicate the respective branch.
In practical solving, the (partial) propositional formula that is obtained by universal expansion can then be handed to a $\sat$ solver~\cite{AB02,B04,LB08,JKMC16}. From a proof-theoretic perspective, calculi like $\mathsf{\forall Exp + Res}$~\cite{JM15} and further generalizations~\cite{BCJ15} combine universal expansion with well-known propositional proof systems such as resolution~\cite{R65}.

\emph{Universal reduction}, the second approach towards eliminating universal variables, processes QBFs in a bottom-up fashion. Here, universal variables can be dropped from a clause if they are not blocked by a co-occurring existential variable that is quantified later. Combining universal reduction with different variants of resolution gave rise to proof systems like $\mathsf{Q}$-$\mathsf{Res}$~\cite{BKF95} and further generalizations~\cite{VG12,BJ12,SS14}. The system $\mathsf{Q}$-$\mathsf{Res}$ in particular models practical $\qsat$ solving based on \emph{conflict-driven clause learning} (CDCL)~\cite{ZM02,GMN09,LB10}. 
Research in proof complexity has revealed that the paradigms of universal expansion and universal reduction are incomparable~\cite{JM15,BCJ14}, that is, their underlying proof systems do not simulate each other. As a consequence, the more sophisticated systems $\mathsf{IR}$-$\mathsf{calc}$ and $\mathsf{IRM}$-$\mathsf{calc}$~\cite{BCJ19} have been introduced to incorporate both paradigms simultaneously.

The \PSPACE-hardness of $\qsat$ has motivated the question of whether structural properties of instances could be employed to achieve tractability. At the foundational level, this study is typically pursued through the lens of \emph{parameterized complexity}~\cite{DF13,CFKLMPPS15}. There, one analyzes the running time of algorithms not only in terms of the input size $|I|$, but also with respect to a specified numerical parameter~$k$, which can be seen as a measure of how ``well-structured'' an instance is. The standard notion of tractability used in this setting is tied to algorithms that run in time $f(k) \cdot |I|^{\oh{1}}$ for some computable function $f$; problems admitting such algorithms are called \emph{fixed-parameter tractable} (FPT).

The most ubiquitous structural measure employed for such investigations is the \emph{treewidth} of the instance, typically measured in terms of its \emph{primal graph}, thus capturing the tree-likeness of its variable-to-variable interactions. But while treewidth has been instrumental in the design of fixed-parameter algorithms for $\sat$~\cite{DP89,F90,S03}, \textsc{CSP}~\cite{SS10}, \textsc{Integer Programming}~\cite{JK15,EGKO19} and many other related problems, $\qsat$ is known to remain \PSPACE-complete even on instances which are ``path-like'' in the sense of having not only constant treewidth but also constant pathwidth~\cite{AO14}---see also the related works~\cite{PV06,FHP20,FGHSO23} for refined lower bounds.
 Intuitively, this is because treewidth does not take into account the prefix of a QBF: it neither distinguishes between existential and universal variables, nor accounts for the order in which they are quantified. Consequently, several weaker variants of pathwidth and treewidth have been proposed to incorporate prefix-sensitive information. These notions split into two competing realms that remarkably echo the two paradigms encountered in proof complexity and practical solving.

The first realm is represented by \emph{prefix pathwidth}~\cite{EGO20}, a variant of pathwidth that places an additional restriction on the order in which variables appear in a path decomposition. More specifically, a variable can only be forgotten when all variables on which it depends have already been introduced in the decomposition. Fixed-parameter tractability with respect to prefix pathwidth is then established in a top-down fashion that essentially eliminates variables from left to right, maintaining at each time step the sets of QBFs that an \emph{existential strategy} could have produced under all universal plays so far. This procedure closely resembles incremental universal expansion with a subsequent explicit instantiation of some existential variables, as will be discussed in \iflong Subsection~\ref{subsec:strategy_ext}\fi\ifshort
Section~\ref{sec:rules}\fi. 
In contrast, \emph{respectful treewidth}~\cite{CD05,AO14} enforces variables to be forgotten in an order that reverses their occurrence in the prefix. Fixed-parameter tractability with respect to this parameter is proved bottom-up by producing variants of $\mathsf{Q}$-$\mathsf{Res}$ proofs. \emph{Dependency treewidth}~\cite{EGO18} generalizes respectful treewidth by incorporating \emph{dependency schemes}~\cite{SS09,LB10,SS14,SS16}, thereby allowing greater flexibility in the order of variable elimination. Here, the produced derivations are certain types of $\qdres$ proofs, where $\D$ is the deployed dependency scheme. Prior work has shown that respectful treewidth and dependency treewidth are incomparable to prefix pathwidth~\cite{EGO18}, see Figure~\ref{fig:tw-variants}.

\smallskip \noindent \textbf{Our Contributions.}\quad
In this work, we introduce \emph{bilateral treewidth}, a new decompositional parameter for $\qsat$ that unifies the two realms discussed above. As our main contribution, we prove that $\qsat$ is fixed-parameter tractable with respect to bilateral treewidth when a suitable tree decomposition is provided as part of the input. Solving $\qsat$ under this parameterization produces derivation sequences of proof-like structure; more specifically, these sequences maintain sets of QBFs  that can either be refined
\begin{description}
    \item[(a)] by a local branching on partial existential strategies that resembles universal expansion~or
    \item[(b)] by employing variants of $\qdres$ steps within the sets.
\end{description}

    Moreover, we show that bilateral treewidth captures tractable instances that are beyond the reach of all previously proposed decompositional parameters for QBF. These separations are obtained by showing that on \emph{parity formulas}~\cite{BCJ15}, both prefix pathwidth and dependency treewidth are high, whereas the bilateral treewidth is a small constant (as witnessed by an easy-to-compute class of corresponding tree decompositions).

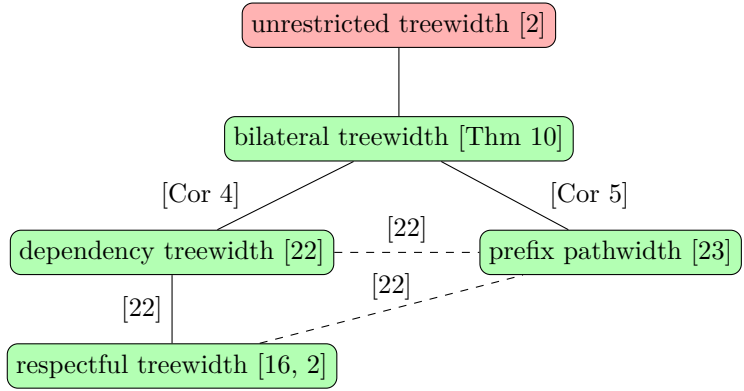
\begin{figure}
    \centering
    \begin{tikzpicture}         \draw (1.7,0.25) node[draw, rounded corners, fill=green!30] (rtw) {respectful treewidth~\cite{CD05,AO14}};
        \draw (7.5,1.75) node[draw, rounded corners, fill=green!30] (ppw) {prefix pathwidth~\cite{EGO20}};
        \draw (1.7,1.75) node[draw, rounded corners, fill=green!30] (dtw) {dependency treewidth~\cite{EGO18}};         \draw (4.7,3.25) node[draw, rounded corners, fill=green!30] (btw) {bilateral treewidth [Thm~\ref{thm:fpt_btw}]};
        \draw (4.7,4.75) node[draw, rounded corners, fill=red!30] (u) {unrestricted treewidth~\cite{AO14}};
        \draw[dashed] (rtw) -- node[above]{\cite{EGO18}} (ppw);
        \draw[dashed] (dtw) -- node[above]{\cite{EGO18}} (ppw);
        \draw (btw) -- node[right]{\phantom{----}[Cor~\ref{cor:ppw_is_btw}]} (ppw);
        \draw (btw) -- node[left]{[Cor~\ref{cor:dtw_is_btw}]\phantom{----}} (dtw);
        \draw (btw) -- (u);
        \draw (rtw) -- node[left]{\cite{EGO18}} (dtw);
    \end{tikzpicture}
    \caption{The hierarchy of decompositional parameters for QBFs. All parameters highlighted in green are tied to decompositions which yield fixed-parameter tractability of $\qsat$.
         Dashed lines between parameters indicate their incomparability and solid lines indicate that the upper parameter is upper-bounded by the lower parameter. In particular, fixed-parameter tractability w.r.t.\ higher parameters implies fixed-parameter tractability w.r.t.\ the lower ones.}
                   \label{fig:tw-variants}
\end{figure}

\section{Preliminaries}\label{subsec:dep}
\label{sec:prelims}
    For a positive integer $i$, we use $[i]$ to denote the set $\{1,\dots,i\}$. We assume familiarity with basic graph terminology.

 \smallskip \noindent \textbf{Quantified Boolean Formulas.}\quad
Let $\Z$ be a countable set of Boolean variables. A \emph{literal} is a Boolean variable $v \in \Z$ or its negation $\overline{v}$, a \emph{clause} $\clause$ is a set of literals, and a \emph{CNF formula}~$\phi$ is a set of clauses. A clause~$\clause$ is \emph{tautological} if $\{v,\overline{v}\}\subseteq \clause$ for some $v \in \Z$ and \emph{empty}, denoted by $\bot$, if it does not contain any literals. 
We write $\var{\phi}$ ($\var{\clause}$, respectively) to denote the set of variables that occur in $\phi$ (in $\clause$, resp.). For each variable $x \in \var{\phi}$, we denote by $\phi_x$ the set of all clauses in $\phi$ that contain the literal $x$.

A \emph{quantified Boolean formula} (\emph{QBF}) $\qbf$ consists of (1) a \emph{prefix} $\Q = \Q_1 \Z_1 \cdots \Q_n \Z_n$, in which all $\Z_i$ are pairwise disjoint finite sets of Boolean variables, $\Q_i \in \{\exists, \forall\}$ for each~$i \in [n]$, and $\Q_i \neq \Q_{i+1}$ for each $i \in [n - 1]$, and (2) a \emph{matrix} $\phi$, which is a CNF formula with $\var{\phi}\subseteq \bigcup_{i \in [n]}\Z_i$.
We define $\var{\Q}=\bigcup_{i \in [n]}\Z_i$ and 
 the set of \emph{existential} (\emph{universal}, respectively) variables of $\Phi$ as $\vare{\Phi}=\bigcup_{i \in [n], \Q_i=\exists}\Z_i$ (as $\vara{\Phi}=\bigcup_{i \in [n], \Q_i=\forall}\Z_i$, resp.), and let $\var{\Phi}=\vare{\Phi}\cup\vara{\Phi}$. 
 For $V \subseteq \var{\Phi}$, we write $\mathcal{Q} \setminus V$ to denote the prefix obtained from $\mathcal{Q}$ after removing all variables in $V$ and their respective quantifiers.

An \emph{assignment} to a set $\Z$ of Boolean variables is a function $\delta: \Z \rightarrow \{0,1\}$. Given some~$\Z' \subseteq \Z$ we denote by $\delta{\upharpoonright}_{\Z'}$ the restriction of $\delta$ to $\Z'$. The set of all assignments to $\Z$ is denoted by $\langle \Z \rangle$. 
Given a QBF $\Phi=\Q.\phi$ and an assignment~$\delta \in \langle Z \rangle$ for some $Z \subseteq \var{\Phi}$, we denote by  $\phi\rvert_{\delta}$ the matrix that is obtained when $\delta$ is applied to $\phi$.

The semantics of a QBF $\qbf$ are neatly described by a \emph{two-player game}. Over the course of a game, the variables of $\qbf$ are assigned to 0 or 1 in the order they appear in the prefix, with the \emph{existential player} (\emph{universal player}, respectively) choosing the value for $\vare{\Phi}$ ($\vara{\Phi}$, resp.).
When the game ends, the players have constructed a total assignment $\delta \in \langle \var{\Phi}\rangle$. The existential player wins if and only if $\phi\rvert_{\delta}\equiv \top$. For every $x \in \vare{\Phi}$, let $L_{\Q}(x)$ be the set of all universal variables that are quantified to the left of $x$ in $\Q$.
An \emph{existential strategy for $\Phi$} is a set $\tau=\{\tau_x \mid x \in \vare{\Phi}\}$ of functions $\tau_x: \langle L_{\Q}(x)\rangle \rightarrow \{0,1\}$ that specify the assignment to be used for each existential variable~$x$ based on the assignment to all universal variables preceding $x$.
An existential strategy~$\tau$ is a \emph{winning existential strategy} if for every universal assignment  $\beta \in \langle \vara{\Phi} \rangle$, we have $\phi\rvert_{\beta \cup \{x \mapsto \tau_x(\beta{\upharpoonright}_{L_{\Q}(x)}) \mid x \in \vare{\Phi}\}}\equiv \top$. 
A QBF is \emph{true} if and only if it admits a winning existential strategy; otherwise we say the QBF is \emph{false}. If a QBF contains no variables, it is considered to be true if its matrix is empty, and false if its matrix contains the empty clause. Two QBFs are \emph{equisatisfiable} if and only if they are either both true or both false.

We remark that the notion of existential strategies defined above can also be lifted to account for the presence of dependency schemes; we elaborate on this in the next subsection.

 \smallskip \noindent \textbf{Dependency Schemes and Posets for QBFs.}\quad
A \emph{partially ordered set} (\emph{poset}) $\mathcal{D}$ is a pair $(V, \dep)$ where $V$ is a set and $\dep$ is a reflexive, antisymmetric, and transitive binary relation over $V$. A \emph{linear extension} of a partially ordered set $\mathcal{D}=(V, \dep)$ is a relation $\leq_O$ over $V$ such that $x \leq_O y$ whenever $x \dep y$ and $(V, \leq_O)$ is a linear order. We use \emph{dependency posets}~\cite{EGO18,EGO20} to provide a general and formal way to speak about the various \emph{dependency schemes} that have been devised for QBFs~\cite{SS09,LB10,SS14,SS16}. 
Not all dependency schemes give rise to sound proof systems, and not all are easy to compute. 
Consequently, we focus solely on the class of \emph{polynomial-time computable permutation dependency schemes}~\cite{SS16}, which can be efficiently deployed in resolution-style proof systems without compromising their soundness. Informally, a permutation dependency scheme assigns to each QBF $\qbf$ a binary relation~$\mathcal{R}$ on its variables such that permuting $\Q$ with respect to any order that is consistent with $\mathcal{R}$ yields an equisatisfiable QBF. 
Based on these dependency schemes, we define the following notion.

Given a QBF $\qbf$, a \emph{(permutation) dependency poset $\mathcal{D}=(\var{\Phi}, \dep)$ of $\Phi$} is a partially ordered set over $\var{\Phi}$ with the following properties:
\begin{enumerate}
    \item for all $x,y \in \var{\Phi}$, if $x \dep y$, then $x=y$ or $x$ is quantified to the left of $y$ in $\Q$ and
    \item for every linear extension $\leq_O$ of $\mathcal{D}$, $\qbf$ is true if and only if $\Q_{\leq_O}. \phi$ is true, where $\Q_{\leq_O}$ is obtained from $\Q$ by rearranging variables together with their respective quantifiers such that their order agrees with $\leq_O$.
\end{enumerate}
 Given a polynomial-time computable permutation dependency scheme, the respective dependency poset can be obtained in polynomial time by computing the transitive closure.
The \emph{trivial dependency scheme} gives rise to the \emph{trivial dependency poset~$\D_{\mathrm{trv}}$ of $\qbf$}, a special type of poset that for all $x,y \in \var{\Phi}$, enforces $x \preceq_{\D_{\mathrm{trv}}} y$, whenever $x=y$ or $x$ is quantified to the left of $y$ in $\Q$. Other examples of polynomial-time computable permutation dependency schemes include $\D_{\mathrm{tree}}$, $\D_{\mathrm{res}}$, $\D_{\mathrm{rrs}}$, and $\D_{\mathrm{std}}$~\cite{L12,SS16}.

Given a QBF $\qbf$ with dependency poset $\D$ we define for each $v \in \var{\Phi}$ the set $\depless{v}=\{v' \mid v' \dep v \text{ and } v' \in \var{\Phi}\}$ of variables that \emph{$v$ depends on}. Due to the properties of $\D$, the definition of existential strategies for $\Phi$ can be refined by replacing $L_{\Q}(x)$ by $\depless{x}$.
By $(\beta, \tau)$, we abbreviate the assignment $\beta \cup \{\tau_x(\beta{\upharpoonright}_{\depless{x}}) \mid x \in \vare{\Phi}\}$ that is created when the existential strategy $\tau$ is applied on a universal assignment  $\beta$. Thus, for a winning existential strategy $\tau$, we have that  $\phi\rvert_{(\beta, \tau)}\equiv \top$ for all $\beta \in \langle \vara{\Phi} \rangle$.

 \smallskip \noindent 
 \textbf{Decompositions, Elimination Orderings, and Cops-and-Robbers.}\quad
    Let $\qbf$ be a QBF. The \emph{primal graph} $G_{\Phi}=(V,E)$ of $\Phi$ 
is the graph whose vertices are the variables of $\Phi$ and where two variables are adjacent if and only if $\phi$ has a clause containing both variables. When $\Q$ is clear from context, we will sometimes use $G_{\phi}$ as shorthand for $G_{\Q.\phi}$.

 A \emph{nice tree decomposition of $\Phi$} is a pair $\mathcal{T}=(\mathbf{T}, \rchi)$, where $\mathbf{T}$ is a tree (whose vertices are called \emph{nodes}) rooted at a node $\rt$ and $\rchi$ is a function that assigns to each node $t$ a set $\rchi(t) \subseteq \var{\Phi}$, called the \emph{bag of $t$}, such that:
\begin{description}
	\item[\namedlabel{itm:T1}{(T1)}] For every $\{v,v'\} \in E(G_{\Phi})$, there is a node $t$ such that $\{v,v'\} \subseteq \rchi(t)$.
	\item[\namedlabel{itm:T2}{(T2)}] For every $v \in \var{\Phi}$, the set of nodes $t$ satisfying $v \in \rchi(t)$ forms a subtree of $\mathbf{T}$.
	\item[\namedlabel{itm:T3}{(T3)}] If $t$ is a \emph{leaf} of $\mathbf{T}$ or $t=\rt$, then $|\rchi(t)| = 0$.
	\item[\namedlabel{itm:T4}{(T4)}] There are only three kinds of non-leaf nodes in $\mathbf{T}$:
	\begin{itemize}
		\item \emph{introduce}: A node $t$ with exactly one child $t'$ and $\rchi(t) = \rchi(t') \cup \{v\}$ for a $v \notin\rchi(t')$. 
		\item \emph{forget}:  A node $t$ with exactly one child $t'$ and $\rchi(t) = \rchi(t') \setminus \{v\}$ for a $v \in \rchi(t')$.  		\item \emph{join}:  A node $t$ with two children $t_1, t_2$ such that $\rchi(t) = \rchi(t_1) = \rchi(t_2)$.
	\end{itemize}
\end{description}
For simplicity, we will refer to nice tree decompositions simply as \emph{tree decompositions} in the following. 
The \emph{width} of a tree decomposition is the size of its largest bag minus 1, and the \emph{treewidth of $\Phi$}
     is the minimum width over all tree decompositions of $\Phi$.

For each node $t$, we let $\mathbf{T}_t$ be the subtree of $\mathbf{T}$ that is rooted in $t$. Further, we let $\mathbf{T}_{\leq t}=\bigcup_{t' \in \mathbf{T}_t} \rchi(t')$ be the set of all vertices that occur in some node in $\mathbf{T}_t$ and let $\mathbf{T}_{<t}=\mathbf{T}_{\leq t} \setminus \rchi(t)$.
For each two distinct nodes $t,t'$ we write $t \lesst t'$ if and only if $t'$ lies on the unique path from~$t$ to $\rt$. For each $v \in \var{\Phi}$, we let  $\forget{v}$ be the unique node $t$ for which $v \in \rchi(t)$, but $v \notin \rchi(t')$ for all nodes $t'$ with $t \lesst t'$. 

In our later sections, we will make use of an alternative characterization of treewidth via elimination orderings. An \emph{elimination ordering of $G_{\Phi}$} is a linear order $\V=(v_1, \dots, v_{n})$ of its vertices, where $n=|V(G_{\Phi})|$. Given $G_{\Phi}$ and an elimination ordering~$\V$ of $G_{\Phi}$, the \emph{fill-in graph} of $G_{\Phi}$ with respect to $\V$ is the unique minimal graph~$G_{\Phi}^{\V}=(V, E)$ such that
$V(G_{\Phi}^{\V})=V(G_{\Phi})$, $E(G_{\Phi}^{\V}) \supseteq E(G_{\Phi})$, and if $1 \leq k < i <j \leq n$ and 
$\{\{v_k,v_i\},\{v_k,v_j\}\} \subseteq E(G_{\Phi}^{\V})$, then $\{v_i,v_j\} \in E(G_{\Phi}^{\V})$---that is, whenever we ``eliminate'' a vertex, we turn its remaining neighbors into a clique. The width of an elimination ordering $\V$ is the maximum number of neighbors of that any vertex in $E(G_{\Phi}^{\V})$ has to its right. 
It is known that a graph has treewidth $k$ if and only if it admits an elimination ordering of width $k$~\cite{K94}. 

 A \emph{path decomposition of $\Phi$} is a tree decomposition~$\mathcal{P}=(\mathbf{P}, \rchi)$, where $\mathbf{P}$ is a path, and the \emph{pathwidth of $\Phi$} is the minimum width over all path decompositions of $\Phi$. Since a path decomposition~$\mathcal{P}$ can be fully characterized by the order of appearance of its bags along~$\lessp$, we can succinctly represent it as a sequence $\mathcal{P}=(P_1, \dots, P_d)$ of bags.
   When comparing to previous work in Section~\ref{sec:comparison}, we will need a directed variant of path decompositions~\cite{B06}. 
A \emph{directed path decomposition} $\mathcal{P}=(P_1, \dots, P_d)$ of a directed graph $D$ is a sequence of subsets of $V(D)$ such that (1) $\bigcup_{i \in [n]}P_i=V(D)$, (2) for each vertex $v\in V(D)$, the set~$I_v=\{i \in [n] \mid v \in P_i\}$ induces an interval, and (3) for each directed edge $(u,v) \in E(D)$, there are $i,j \in [n]$ with $i \leq j$ such that $u \in P_i$ and $v \in P_j$. The \emph{width} of $\mathcal{P}$ and the \emph{directed pathwidth of $D$}, denoted by $\dpw{D}$, are then defined analogously as for path- and treewidth. 
  
 Directed pathwidth is closely related to
a directed version of a \emph{cops-and-robber game}~\cite{NW83,ST93} which will facilitate our lower-bound arguments in Lemma~\ref{lem:ppw_parity}---we refer to Bar\'{a}t~\cite{B06} for a full description. Essentially, when given a directed graph $D$, the \emph{robber} occupies one of its vertices $v$ and can run at infinite speed to any other vertex that is reachable from $v$ via a directed path on which no vertex is occupied by a cop. The \emph{cops} also occupy one vertex each, and in each move the cops player can lift a cop and announce a new location for it, which can be any vertex of the graph. The game is played in rounds as follows: the cops lift and announce a move, the robber responds by moving, and the cops complete the move.
  Crucially, 
in this version the cops cannot see which vertex the robber is occupying. We say that the cops \emph{capture} the robber if a cop moves to the vertex $v$ that is currently occupied by the robber and every vertex that the robber can reach from $v$ is also occupied by a cop. The \emph{cop number of $D$}, denoted by $\cn{D}$, is the minimum number of cops that suffice to guarantee a capture of the robber on $D$. It is known that $\dpw{D}\leq k-1$ implies $\cn{D}\leq k$~\cite{B06}.

\section{Derivation Rules: Q($\mathcal{D}$)-Resolution and Strategy Extension}\label{sec:rules}
In this section, we introduce the derivation rules that will be used in Section~\ref{sec:fpt_proof} to prove fixed-parameter tractability of $\qsat$ with respect to bilateral treewidth. Our first two rules, \emph{Resolution} and \emph{Reduction}, are tightly linked to the proof system $\mathsf{Q}$-$\mathsf{Res}$ that was introduced by Kleine B\"{u}ning~\cite{BKF95} and its dependency-aware variants~\cite{SS14}. Our adaptation of these rules, however, eliminates variables in an exhaustive way that is more akin to \emph{multi-resolution proofs}~\cite{CS00,PV04,AO14}. Crucially, in contrast to these prior works, we handle tautologies differently to allow for the sound application of an additional rule that we call \emph{Strategy Extension}.

\ifshort
\smallskip \noindent \textbf{Eliminating Existential Variables via Resolution.}\quad
\fi
\iflong
\subsection{Eliminating Existential Variables via Resolution}
\fi
Let $\qbf$ be a QBF that does not contain any tautological clauses. We define the following rule. 
\vspace{-0.2cm}
\begin{prooftree}\label{rule:resolution}
   \AxiomC{$\phi$}
   \LeftLabel{\emph{Resolution over} $x \in \vare{\Phi}:\quad$}
   \UnaryInfC{$\res{\phi}{x}=\phi \setminus (\phi_{x} \cup \phi_{\overline{x}}) \cup \R{\phi}{x}$,}
   \end{prooftree}  
where $\R{\phi}{x}=\{\clause \mid \clause \text{ is not tautological and there exist } \clause_1\in \phi_x, \clause_2\in \phi_{\overline{x}} \text{ such that } \clause=(\clause_1 \setminus \{x\}) \cup (\clause_2 \setminus \{\overline{x}\})\}$.
   We say that we \emph{resolve over $x$ in $\phi$} whenever we compute $\res{\phi}{x}$ and call $x$ the \emph{pivot variable} of this resolution step. \ifshort The correctness of the following two lemmas essentially follows from~\cite{BKF95}. See the long version of this paper for respective proofs.\fi

 \begin{lemma}
 \label{lem:res_correctness}
 Let $\qbf$ be a QBF that contains no tautological clause, let $x \in \vare{\Phi}$, and $\mathcal{D}$ be a dependency poset of $\Phi$. Suppose that for all $u\in \vara{\Phi}$ with $x \dep u$, there is no clause $ \clause \in \phi$ such that $\{x,u\} \subseteq \var{\clause}$. Then $\Phi$ is true if and only if $\mathcal{Q} \setminus \{x\}. \res{\phi}{x}$ is true. Moreover, no clause in $\res{\phi}{x}$ is tautological or contains a literal of~$x$.
   \end{lemma}
 \iflong
 \begin{proof}
    The second statement is readily observed  since neither $\phi$ nor $\R{\phi}{x}$ contain any tautological clauses and no clause in $\R{\phi}{x}$ contains any literal in $x$.
    
     Towards the first statement, suppose that $\qbf$ is true, that is, there exists a winning existential strategy $\tau$ for $\Phi$. We will prove that $\tau\setminus \{\tau_x\}$ is a winning existential strategy for $\mathcal{Q} \setminus \{x\}.\res{\phi}{x}$.
     Assume that there exists an assignment $\beta \in \langle \vara{\Phi}\rangle$ such that some clause $\clause \in \res{\phi}{x}$ is falsified by $(\beta, \tau\setminus \{\tau_x\})$.
     Clearly, $\clause \in \R{\phi}{x}$, as otherwise $\clause \in \phi$ and $\tau$ would not be winning. Thus, there exist $\clause_1 \in \phi_{x}$ and $\clause_2 \in \phi_{\overline{x}}$ such that $\clause=(\clause_1 \setminus \{x\}) \cup (\clause_2 \setminus \{\overline{x}\})$ is falsified by $(\beta,\tau\setminus \{\tau_x\})$. Assume w.l.o.g. that $\tau_x(\beta{\upharpoonright}_{\depless{x}})=1$.
     Then, since $\clause_2 \in \phi$ is satisfied by $(\beta,\tau)$, we have that $\clause_2 \setminus \{\overline{x}\}$ is satisfied by $(\beta,\tau\setminus \{\tau_x\})$. But then $\clause$ is also satisfied by $(\beta,\tau\setminus \{\tau_x\})$, which contradicts our assumption.

    Suppose now $\qbf$ is false. Then for each partial existential strategy $\tau$ on $X=(\vare{\Phi} \cap \depless{x}) \setminus \{x\}$, there exists an assignment $\beta_{\tau} \in \langle Y\rangle$ with $Y = \vara{\Phi} \cap \depless{x}$, such that $\Phi^{\beta_{\tau}}=\Q \setminus (X \cup Y).\phi^{\beta_{\tau}}$ with $\phi^{\beta_{\tau}}=\phi \rvert_{(\beta_{\tau},\tau)}$ is false. Fix such a $\tau$ arbitrarily and let $\res{\phi}{x}^{\beta_{\tau}}:=\res{\phi}{x}\rvert_{(\beta_{\tau}, \tau)}$. 
     We will prove that $\Q \setminus (X \cup Y \cup \{x\}). \res{\phi}{x}^{\beta_{\tau}}$ is false, which entails that $\Q \setminus  \{x\}. \res{\phi}{x}$ is false since $\tau$ was chosen arbitrarily.
    If for all existential strategies $\tau'$ on $\vare{\Phi^{\beta_{\tau}}}$, there exists an assignment $\beta'_{\tau'} \in \langle \vara{\Phi^{\beta_{\tau}}}\rangle$ such that $(\beta'_{\tau'}, \tau') $ falsifies a clause $\clause \in \phi^{\beta_{\tau}} \setminus ((\phi^{\beta_{\tau}})_x \cup (\phi^{\beta_{\tau}})_{\overline{x}})$, then this clause will also be present in the resolvent, i.e., $\clause \in \res{\phi}{x}^{\beta_{\tau}}$ and hence $\Q \setminus (X \cup Y). \res{\phi}{x}^{\beta_{\tau}}$ is false.
    Let us therefore focus on all existential strategies $\tau'$ for which all clauses in $\phi^{\beta_{\tau}} \setminus ((\phi^{\beta_{\tau}})_x \cup (\phi^{\beta_{\tau}})_{\overline{x}})$ are satisfied by $(\beta'_{\tau'}, \tau') $ for all $\beta'_{\tau'}$. Since $\Phi^{\beta_{\tau}}$ is false, we can find for each such $\tau'$ some $\beta'_{\tau'}$ such that some clause in $ (\phi^{\beta_{\tau}})_x \cup (\phi^{\beta_{\tau}})_{\overline{x}}$ is falsified by $(\beta'_{\tau'}, \tau')$. Recall from the statement of the Lemma that no clause in $ (\phi^{\beta_{\tau}})_x \cup (\phi^{\beta_{\tau}})_{\overline{x}}$ contains any universal variable that depends on~$x$. Hence there exists some $\beta'_{\tau'}$ such that $(\beta'_{\tau'}, \tau')$ falsifies $\clause_1 \setminus \{x\}$ for some clause $\clause_1 \in (\phi^{\beta_{\tau}})_x$ and also falsifies $\clause_2 \setminus \{\overline{x}\}$ for some clause $\clause_2 \in (\phi^{\beta_{\tau}})_{\{\overline{x}\}}$, as $\beta'_{\tau'}$ has to guarantee a falsified clause for each possible play on $x$.
         But then, the clause $\clause = (\clause_1 \setminus \{x\} \cup \clause_2 \setminus \{\overline{x}\})\in \res{\phi}{x}^{\beta_{\tau}}$ is falsified by $(\beta'_{\tau'}, \tau' \setminus \{\tau_x\})$.

                                                             \end{proof}
\fi

\ifshort
\smallskip \noindent \textbf{Eliminating Universal Variables via Reduction.}\quad
\fi
\iflong
\subsection{Eliminating Universal Variables via Reduction}
\fi
Let $\qbf$ be a QBF that does not contain any tautological clauses. We define the following rule. 
\vspace{-0.2cm}
\begin{prooftree}\label{rule:reduction}
  \AxiomC{$\phi$}
  \LeftLabel{\emph{Reduction of} $u \in \vara{\Phi}:\quad$}
  \UnaryInfC{$\red{\phi}{u}=\{\clause \setminus \{u, \overline{u}\} \mid \clause \in \phi\}$.}
  \end{prooftree}   
We say that we \emph{reduce $u$ from $\phi$} whenever we compute $\red{\phi}{u}$. Note that if $\phi$ contains the clause $\{u\}$ or $\{\overline{u}\}$, then $\red{\phi}{u}$ contains the empty clause, indicating a false instance.

\begin{lemma}
\label{lem:red_correctness}
Let $\qbf$ be a QBF that contains no tautological clause, let $u \in \vara{\Phi}$, and $\mathcal{D}$ be a dependency poset of $\Phi$. Suppose that for all $x\in \vare{\Phi}$ with $u \dep x$, there is no clause $ \clause \in \phi$ such that $\{x,u\} \subseteq \var{\clause}$.
 Then $\Phi$ is true if and only if $\mathcal{Q} \setminus \{u\}. \red{\phi}{u}$ is true. 
Moreover, no clause in $\red{\phi}{u}$ is tautological or contains a literal of~$u$.
\end{lemma}
\iflong
\begin{proof}
    The latter claim follows from the fact that $\phi$ does not contain any tautological clauses and that both $u$ and $\overline{u}$ are exhaustively eliminated.
    
    Towards the former claim, suppose that $\Phi$ is true, that is, there exists a winning existential strategy $\tau$ for $\Phi$. Let $\tau'$ be the existential strategy that copies $\tau$ on all $x \in \vare{\Phi}$ for which $u \not\in \depless{x}$ and for all remaining $x$ copies $\tau$ on the partial assignment $u=0$. We claim that~$\tau'$ is a winning existential strategy for $\Q \setminus \{u\}.\red{\phi}{u}$.
    Suppose not, that is, there exists a universal assignment $\beta \in \langle \vara{\Phi} \setminus \{u\}\rangle$ such that some clause $\clause \in \red{\phi}{u}$ is falsified by $(\beta, \tau')$. If $\clause \in \phi$, then $\clause$ is falsified by $(\beta \cup \{u=0\}, \tau)$, which contradicts the fact that $\tau$ is winning on $\Phi$.
    Since by assumption, $(\clause \cup \{u, \overline{u}\}) \not\in \phi$, we have that $(\clause \cup \{u\}) \in \phi$ or $(\clause \cup \{\overline{u}\} )\in \phi$. In the first case, $\phi$ is falsified by $(\beta \cup \{u = 0\}, \tau)$. In the latter, $\phi$  is falsified by $(\beta \cup \{u = 1\}, \tau)$, since by assumption, $\clause$ does not contain any $x \in \vare{\Phi}$ for which $u \in \depless{x}$ and therefore, $\tau$ and $\tau'$ agree on all variables in $\clause$. Both cases contradict the fact that $\tau$ is a winning existential strategy.

    For the other direction, note that all clauses in $\red{\phi}{u}$ are subclauses of clauses in~$\phi$. Therefore, every winning existential strategy $\tau'$ for $\Q\setminus \{u\}.\red{\phi}{u}$ can be extended to a winning existential strategy $\tau$ for $\qbf$ that copies $\tau'$ on all existential variables, irrespective of the assignment to $u$. 
\end{proof}
\fi

\ifshort
\smallskip \noindent \textbf{Eliminating Variables via Strategy Extension.}\quad
\fi
\iflong
 \subsection{Eliminating Variables via Strategy Extension}\label{subsec:strategy_ext}
 \fi
 \emph{Strategy Extension} is a rule based on a branching technique deployed in~\cite{EGO20}. In contrast to resolution and reduction, this rule operates on multiple QBFs simultaneously. Intuitively, these QBFs correspond to every possible formula that a specific partial existential strategy could have produced when responding to every possible universal play thus far. We will show in Section~\ref{sec:fpt_proof} that the number of strategies and QBFs that need to be considered in parallel at each time step is bounded by a function of the bilateral treewidth.
 
Let $\pi=\{\psi_{1}, \dots, \psi_{\ell}\}$ be a set of matrices that all share the same prefix $\Q$, that is, $\bigcup_{j \in [\ell]}\var{\psi_j} \subseteq \var{\Q}$ and let no $\psi_j$ contain a tautological clause.
For some $v \in \var{\Q}$ up to which we wish to perform strategy extension, we define $\B=\langle \vara{\Q} \cap\depless{v} \rangle$ to be the set of all universal assignments to the set of variables that $v$ depends on. Moreover, we let~$\A$ be the set of all (partial) existential strategies restricted to the variable set $\vare{\Q} \cap \depless{v}$, that is, each such strategy specifies a function $\tau_x: \B{\upharpoonright}_{\depless{x}} \rightarrow \{0,1\}$ for all $x \in \vare{\Q} \cap \depless{v}$.
   We define the following rule.
\vspace{-0.2cm}
\begin{prooftree}\label{rule:strategy_ext}
  \AxiomC{$\pi=\{\psi_{1}, \dots, \psi_{\ell}\}$}
  \LeftLabel{\emph{Strategy Extension up to} $v \in \var{\Q}:\quad$}
  \UnaryInfC{$ \strext{\pi}{v}= \{\pi_{(\tau^{1}, \dots, \tau^{\ell})} \mid (\tau^{1}, \dots, \tau^{\ell}) \in \A^{\ell}\}$,}
     \end{prooftree}  
    where $\pi_{(\tau^{1}, \dots, \tau^{\ell})}=\bigcup_{j \in [\ell]}\pi_{\psi_{j}}^{\tau^{j}}$ with $\pi^{\tau^{j}}_{\psi_{j}}=\{\psi_{j}\rvert_{(\beta, \tau^j)} \mid \beta \in \B\}$. Recall from the previous section that $(\beta, \tau^j)$ abbreviates the assignment $\beta \cup \{\tau^j_x(\beta{\upharpoonright}_{\depless{x}}) \mid x \in (\vare{\Q} \cap \depless{v})\}$.
     Essentially, $\strext{\pi}{v}$ enumerates all $\ell$-tuples of partial strategies---one tuple entry for each of the $\ell$ initial QBFs---and for each such tuple, outputs the set of all QBFs that can arise from the initial QBF by having these strategies respond to universal plays. Note that while a single tuple from $\A^{\ell}$ results in a set of QBF instances, $\strext{\pi}{v}$ is a set of such sets.

     \begin{lemma}
\label{lem:ext_correctness}
    Let $\pi=\{\psi_{1}, \dots, \psi_{\ell}\}$ be a set of matrices  that share a prefix $\Q$ and contain no tautological clauses, let $v \in \var{\Q}$ and let $\Q'=\Q \setminus \depless{v}$. 
        Then all QBFs in $\{\Q.\psi_j~|~j\in [\ell] \}$ are true if and only if there
             exists some $\pi' \in \strext{\pi}{v}$ such that all QBFs in $\{\Q'.\psi'~|~\psi' \in \pi' \}$ are true.
         Moreover, all matrices in all $\pi' \in \strext{\pi}{v}$ neither contain a tautological clause nor a clause that contains a literal of any variable in $\depless{v}$.
\end{lemma}
\iflong
\begin{proof}
    The latter claim follows from the fact that each clause which occurs in a matrix in some $\pi' \in \strext{\pi}{v}$ is a subclause of a clause in some $\psi_j\in \pi$, which does not contain any tautological clauses by definition. Moreover, each variable in $\depless{v}$ either gets assigned universally (if it is contained in $\B$) or is set by a partial strategy in $\A$.

    Towards the former claim, let $\B$ be defined as above.
    Suppose all QBFs in $\{\Q.\psi_j~|~j\in [\ell] \}$ are true, that is, for every $\Q. \psi_j$, there exists an existential strategy $\rho^{j}=\{\rho^j_{x} \mid x \in \vare{\Q. \psi_j}\}$ with $ \rho^j_{x}: \langle \B {\upharpoonright}_{\depless{x}} \rangle \rightarrow \{0,1\}$, such that for all $\beta' \in \langle \vara{\Q. \psi_j} \rangle$ it holds that $\psi_j \rvert_{(\beta', \rho^j)}\equiv \top$. For each $j \in [\ell]$, let $\tau^j=\{\tau^j_x \mid \tau^j_x = \rho^j_x, x \in (\vare{\Q.\psi_j} \cap \depless{v}) \}$ be the restriction of $\rho^j$ to existential variables that $v$ depends on. Then for all $j \in [\ell]$ and all $\beta \in \B$, the QBF $\Q'.\psi_j \rvert_{(\beta, \tau^j)}$ is true. Note that by definition, $\pi_{(\tau^1, \dots, \tau^{\ell})}=\bigcup_{j \in [\ell]} \{\psi_j \rvert_{(\beta, \tau^j)} \mid \beta \in \B\}$. 
    Consequently, $\Q'. \psi'$ is true for all $\psi' \in \pi_{(\tau^1, \dots, \tau^{\ell})}$.

     For the reverse direction, suppose $\Q. \psi_j$ is false for some $j \in[\ell]$, that is, for every existential strategy $\rho^{j}=\{\rho^j_{x} \mid x \in \vare{\Q. \psi_j}\}$ with $ \rho^j_{x}: \langle \B {\upharpoonright}_{\depless{x}} \rangle \rightarrow \{0,1\}$, there is a $\beta_{\rho^j} \in \langle \vara{\Q. \psi_j} \rangle$ with $\psi_j \rvert_{(\beta_{\rho^j}, \rho^j)} \equiv \bot$. For each $\tau \in \A$, find some $\rho^j$ such that for its restriction $\tau^j$ (defined as above), it holds that $\tau^j=\tau$. Then, for every $\tau \in \A$, there is a partial universal assignment (namely $\beta= \beta_{\rho^j}{\upharpoonright}_{\depless{v}}$) such that $\Q'.\psi_j \rvert_{(\beta, \tau)}$ is false. Consequently, for each choice $\tau \in \A$, there is some $\psi' \in \pi_{(\tau^1, \dots, \tau^{j-1}, \tau, \tau^{j+1}, \dots, \tau^{\ell})}$ for which $\Q'.\psi'$ is false.
\end{proof}

\subsection{Relation of Strategy Extension to Universal Expansion}
    As noted, the first two rules are global variants of the resolution and reduction rule that are employed by various QBF proof systems, such as $\qdres$. Below, we show that strategy extension can likewise be linked to an existing rule, known as \emph{universal expansion}~\cite{JM15}. This rule is formally defined for a QBF $\qbf$ as follows.
 \fi

 \ifshort
While the first two rules are global variants of the resolution and reduction rule that are employed by, e.g., $\qdres$, strategy extension can also be linked to an existing rule---specifically, to a rule known as \emph{universal expansion}~\cite{JM15}. 
 \fi
\vspace{-0.2cm}
\begin{prooftree}
  \AxiomC{$\phi$}
  \LeftLabel{\emph{Universal Expansion up to} $v \in \vara{\Phi}:\quad$}
  \UnaryInfC{$ \expa{\phi}{v}=\bigwedge_{\beta \in \B} \phi\rvert_{\beta \cup \{x \mapsto x_{\beta{\upharpoonright}_{L_{\Q}(x)}} \mid x \in \vare{\phi}\}}$,}
  \end{prooftree} 
  where $\B$ is defined as above and $x_{\beta{\upharpoonright}_{L_{\Q}(x)}}$ is a new variable that is \emph{annotated} with the universal assignment $\beta{\upharpoonright}_{L_{\Q}(x)}$. The relation between strategy extension and universal expansion can be understood through two key observations.
  \begin{itemize}
      \item Both rules branch on all assignments in $\B$, but employ different formalizations: 
       strategy extension handles the formulas that arise from different universal assignments separately, whereas universal expansion maintains a single formula and uses annotated copies of variables to account for the universal branching. This difference is purely syntactic.
       \item Strategy extension goes beyond universal expansion by explicitly assigning all existential variables up to $v$, whereas universal expansion just keeps their annotated counterparts.
  \end{itemize}
Based on these observations, strategy extension up to a variable $v$ can be simulated by first performing universal expansion up to $v$, and then instantiating all annotated existential variables up to $v$. 
 \iflong
To simplify the exposition, we consider only the case where strategy extension is applied to a single matrix. 
\begin{longobservation}
    Let $\qbf$ be a QBF, $v \in \vara{\Phi}$, and $\E=\vare{\Q} \cap \depless{v}$, and let $\A$ and $\Q'$ be defined as above.
    There is a $\pi' \in \strext{\{\phi\}}{v}$ such that all QBFs in $\{\Q'.\psi'~|~\psi' \in \pi' \}$ are true if and only if 
    $\Q''.\bigvee_{\tau' \in \A'}\expa{\phi}{v}\rvert_{\tau'}$ is true, where for each $\tau \in \A$, the set $\A'$ contains the constant partial existential strategy $\tau'=\{\tau_{x_{\sigma}}'\mid x_{\sigma}\in \var{\expa{\phi}{v}}, x \in \E\}$ with $\tau_{x_{\sigma}}'=\tau_x(\sigma{\upharpoonright}_{\depless{x}})$, that is, $\tau'$ sets each annotated variable $x_{\sigma}$ to whichever value the strategy~$\tau$ chooses for $x$ when given the partial universal assignment $\sigma{\upharpoonright}_{\depless{x}}$. Here, $\Q''$ is obtained from $\Q'$ by replacing each $x \in \vare{\Q'}$ with the set of variables $\{x_{\beta{\upharpoonright}_{L_{\Q}(x)}} \mid \beta \in \B\}$.
               \end{longobservation}
\begin{proof}
    Suppose, there exists a $\pi' \in \strext{\{\phi\}}{v}=\{\{ \phi\rvert_{(\beta, \tau)} \mid \beta \in \B\}\mid \tau \in \A\}$ such that all QBFs in $\{\Q'.\psi'~|~\psi' \in \pi' \}$ are true. In other words, there is a $\tau \in \A$ such that $\Q'.\phi\rvert_{(\beta,\tau)}=\Q'.\phi\rvert_{\beta \cup \{\tau_x(\beta{\upharpoonright}_{\depless{x}}) \mid x \in \E\}}$ is true for all $\beta \in \B$. We can express this property as a single formula by---for each formula that arises under a certain $\beta \in \B$---annotating all existential variables in this formula with $\beta$. We can now merge the individual formulas, since the sets of existential variables that they contain are all pairwise disjoint.
    Formally, we conclude from the above that $\Phi' := \Q''.\bigvee_{\tau \in \A}\bigwedge_{\beta \in \B} \phi\rvert_{(\beta,\tau) \cup \{x \mapsto x_{\beta{\upharpoonright}_{L_{\Q}(x)}} \mid x\in \vare{\Q'}\}}$ is true. Observe that $\Phi'$ does not contain any variable in $\E$ since they are assigned by each $\tau\in \A$. Hence, the same formula could have been produced by applying universal expansion up to $v$ first and then assigning each annotated variable $x_{\sigma}$ that was derived from an $x \in \E$ according to $\tau'$, that is, by treating its annotation as an input to the strategy function.
    Formally, we have $\Phi' \equiv \Q''. \bigvee_{\tau' \in \A'}\expa{\phi}{v}\rvert_{\tau'}$.
    The reverse direction can be proved analogously.
\end{proof}
\fi

\section{Tree Decompositions for QBFs}
Recall that $\qsat$ is \pspace-complete even on instances of constant treewidth~\cite{AO14}.
In this section, we adjust the general notion of tree decompositions to enable fixed-parameter tractability of $\qsat$ via the careful application of the rules introduced in Section~\ref{sec:rules}.
 We call the resulting parameter bilateral treewidth and demonstrate that it strictly dominates all previous measures of tree-likeness considered for QBFs.
 \subsection{Trunk-Aligned Tree Decompositions and Bilateral Treewidth}
Let $\qbf$ be a QBF with dependency poset $\D$ and let $\mathcal{T}=(\mathbf{T},\rchi)$ be a tree decomposition of $\Phi$. 
For some $\ell \in \mathbb{N}$, let $\trunk=(t_0, \dots, t_{\ell})$ be a designated sequence of nodes $t_0, \dots, t_{\ell}$ such that $t_0$ is a leaf, $t_{\ell}=\rt$, and $t_{i-1}$ is the child of $t_i$ for all $i \in [\ell]$, that is, $\trunk$ is a leaf-to-root path. We refer to $\trunk$ as the \emph{trunk} of $\tree$. 
We call $\tree$ a \emph{trunk-aligned tree decomposition} of~$\Phi$ with respect to $\D$ if it satisfies the following property imposed on the order in which variables appear in bags:
For all $u \in \var{\Phi}$, we have that
\begin{description}
    \item[\namedlabel{itm:P1}{(P1)}] for all $v \in \var{\Phi}$ with $u \in \depless{v}$, it is the case that $v\notin\rchi(\forget{u})$ or
               \item[\namedlabel{itm:P2}{(P2)}] $\forget{u} \in \trunk$ and for all $v \in \depless{u}$,  it is the case that $v\in \mathbf{T}_{\leq\forget{u}}$.
           \end{description}
         Informally, Property~\ref{itm:P1} requires all variables that depend on $u$ to avoid the bag where $u$ is forgotten---that is, they are forgotten inside $\mathbf{T}_u$ or only occur outside of that subtree. Inside the trunk, Property~\ref{itm:P2} allows for an alternative: all variables that $u$ depends on are introduced ``before'' $u$ is forgotten.
The \emph{width} of a trunk-aligned tree decomposition is defined analogously as for standard tree decompositions.
The \emph{bilateral treewidth} of $\Phi$ with respect to $\D$, denoted by $\btw{\Phi, \D}$, is the minimal width over all trunk-aligned tree decompositions of $\Phi$ with respect to $\D$.

In Section~\ref{sec:fpt_proof}, we will use trunk-aligned tree decompositions to construct proof-like derivation sequences.
Intuitively, whenever a variable is forgotten in a trunk-aligned tree decomposition, we will eliminate it from all QBFs that we currently maintain in our derivation sequence by using one of the rules from Section~\ref{sec:rules}. As we will prove in Section~\ref{sec:fpt_proof}, Property~\ref{itm:P1} ensures that we invoke the reduction rule (the resolution rule, respectively) only on universal (existential, resp.) variables that are not blocked by a co-occurring existential (universal, resp.) variable that is quantified later.  Similarly, Property~\ref{itm:P2} guarantees that eliminating variables via strategy extension is tractable.
As a result, our approach not only solves $\qsat$, but also produces concise proofs; previously, dependency treewidth was known to admit such proof extractions, whereas prefix pathwidth was not.
 Before proceeding to Section~\ref{sec:fpt_proof}, we briefly discuss these measures and their relationship to bilateral treewidth.

\subsection{Comparison of Decompositional Parameters for QSAT}\label{sec:comparison}
For the following considerations, it will be useful to recall Figure~\ref{fig:tw-variants}. Our aim here is not only to provide an overview of previously introduced measures for the tree-likeness of QBFs, but also to establish that bilateral treewidth generalizes all of them.
 
\smallskip \noindent \textbf{Dependency Tree Decompositions.}\quad
Dependency tree decompositions were introduced in~\cite{EGO18} and are defined as follows:
Let $\qbf$ be a QBF with dependency poset $\D$ and $\mathcal{T}=(\mathbf{T},\rchi)$ be a tree decomposition of $\Phi$.
We call $\tree$ a \emph{dependency tree decomposition} of $(\Phi, \D)$ if it satisfies the following property: For all $u,v \in \var{\Phi}$ with $u \in \depless{v}$, it is \textbf{not} the case that $\forget{u}\lesst \forget{v}$. Note that this condition resembles~\ref{itm:P1}. However, the condition imposed on dependency tree decompositions is more restrictive since $v \in \rchi(\forget{u})$ implies that $\forget{u} \lesst \forget{v}$.
  Therefore, we conclude the following.
   \begin{corollary}\label{cor:dtw_is_btw}
    Every dependency tree decomposition~$\mathcal{T}=(\mathbf{T},\rchi)$ is a trunk-aligned tree decomposition where $\trunk$ is an arbitrary leaf-to-root-path in $\tree$.
\end{corollary}
The \emph{dependency treewidth} of $\Phi$ with respect to a dependency poset $\D$, denoted by $\dtw{\Phi, \D}$, is the minimal width of all dependency tree decompositions of $(\Phi, \D)$. 
   
 Dependency treewidth can be characterized in terms of elimination orderings. More precisely, we define the \emph{$\D$-elimination ordering} of $G_{\Phi}$ to be an elimination ordering of $G_{\Phi}$ that is a linear extension of the reverse of $\D$. It is known that $\dtw{\Phi, \D}$ is equal to the minimal width among all $\D$-elimination orderings of $G_{\Phi}$~\cite{EGO18}.
 
\smallskip \noindent \textbf{Respectful Tree Decompositions.}\quad
Chen and Dalmau~\cite{CD05} introduced a concept that is equivalent to respectful tree decompositions, but is phrased in terms of elimination orderings.
 Atserias and Oliva~\cite{AO14} re-stated it in terms of tree decompositions and coined the name. Having introduced the notion of dependency tree decompositions first---although respectful tree decompositions appeared earlier in the literature---we can now neatly define the latter as a special case of the former: A \emph{respectful tree decomposition} is a dependency tree decomposition with respect to the trivial dependency poset $\D_{\mathrm{trv}}$. The \emph{respectful treewidth} of some $\Phi$ is the minimal width of all respectful tree decompositions of $\Phi$.
 
\smallskip \noindent \textbf{Prefix Path Decompositions.}\quad
Prefix path decompositions~\cite{EGO20} are conceptually different from the previous two measures.
Let $\Phi$ be a QBF with dependency poset $\D$ and $\mathcal{P}=(\mathbf{P},\rchi)$ be a path decomposition of $\Phi$.
We call $\mathcal{P}$ a \emph{prefix path decomposition} of $(\Phi, \D)$ if for all $u \in \var{\Phi}$ and all $v \in \depless{u}$, we have $v \in \mathbf{P}_{\leq\forget{u}}$.
The resemblance to~\ref{itm:P2} is evident. Consequently, we obtain:
\begin{corollary}\label{cor:ppw_is_btw}
    Every prefix path decomposition~$\mathcal{P}$ with $\mathcal{P}=(P_1, \dots, P_d)$ is a trunk-aligned tree decomposition with $\trunk=(P_1, \dots, P_d)$.
\end{corollary}
The \emph{prefix pathwidth} of $\Phi$ with respect to dependency poset $\D$, denoted by $\ppw{\Phi, \D}$, is the minimal width of all prefix path decompositions of $(\Phi, \D)$.
Eiben, Ganian, and Ordyniak~\cite{EGO20} provided a characterization of prefix pathwidth in terms of the directed pathwidth of a graph 
$D_{\Phi}^{\D}$ which can be obtained from $G_{\Phi}$ as follows. We replace every edge~$\{u,v\}$ with the two directed edges $(u,v)$ and $(v,u)$, and add the directed edge $(u,v)$ for each two vertices $u,v \in V(G_{\Phi})$ with $u \dep v$. Then, $\ppw{\Phi, \D}=\dpw{\Phi,\D}$, where $\dpw{\Phi, \D}$ is the directed pathwidth of $D^{\D}_{\Phi}$ as introduced at the end of Section~\ref{sec:prelims}.

\smallskip \noindent \textbf{The Bilateral Treewidth of QParity is Small.}\quad
While Corollaries~\ref{cor:dtw_is_btw} and~\ref{cor:ppw_is_btw} suffice to show that bilateral treewidth subsumes dependency treewidth and prefix pathwidth, it is not yet clear whether it is genuinely stronger, that is, whether it captures tractable instances which the previous parameters did not. We answer this question in the positive.

 \begin{theorem}
    \label{thm:parity}
    For each $n \in \mathbb{N}$ with $n\geq 2$, there exists a QBF $\Phi$ with $2n+1$ variables such that $\dtw{\Phi, \Dtrv}\geq n+1$ and $\ppw{\Phi, \Dtrv}\geq n$, but $\btw{\Phi, \Dtrv}=2$. \end{theorem}

 We establish Theorem~\ref{thm:parity} by showing that bilateral treewidth supersedes the previous notions on the well-studied QBF family $\qpar$, which is known to separate $\mathsf{Q}$-$\mathsf{Res}$ and $\forall\mathsf{Exp}+\mathsf{Res}$~\cite{BCJ15}.
 For $n \in \mathbb{N}$ with $n\geq 2$, the $n$-th formula in this family is defined as follows.
 \begin{align*}
    \qpar:=\exists x_1 \cdots x_n \forall u \exists z_n \cdots z_1 . \mathrm{eq}(x_1,z_1)  \land \mathrm{eq}(u, z_n)  \land \bigwedge_{i \in [n-1]}  \mathrm{xor}(z_{i+1},x_{i+1},z_i), 
\end{align*}
 where $\mathrm{eq}(\circ,\star) := (\circ \lor \overline{\star}) \land (\overline{\circ} \lor \star)$ and
 $\mathrm{xor}(z_{i+1},x_{i+1},z_i) :=  (\overline{z_{i+1}} \lor x_{i+1} \lor z_i) \land (z_{i+1} \lor \overline{x_{i+1}} \lor z_i) \land (z_{i+1} \lor x_{i+1} \lor \overline{z_i}) \land (\overline{z_{i+1}} \lor \overline{x_{i+1}} \lor \overline{z_i})$. We let $\mathcal{X}_n=\{x_1, \dots, x_n\}$ and $\mathcal{Z}_n=\{z_1, \dots, z_n\}$.

 As a visual reference, Figure~\ref{fig:parity} shows the primal graph of $\textsc{QParity}_5$. We proceed to prove each statement in Theorem~\ref{thm:parity} individually.
 \begin{figure}[h!]
\captionsetup[subfigure]{font=footnotesize, justification=centering}
\centering
\subcaptionbox{The primal graph $G_{\textsc{QParity}_5}\!.$\label{fig:parity}}[.5\textwidth]{ \begin{tikzpicture}
    \node[circle, fill, inner sep=0pt, minimum size=5pt, label=below:{$z_1$}](z1) at (0,0) {};
    \node[circle, fill, inner sep=0pt, minimum size=5pt, label=above:{$x_1$}](x1) at (0,1) {};
    \node[circle, fill, inner sep=0pt, minimum size=5pt, label=below:{$z_2$}](z2) at (1,0) {};
    \node[circle, fill, inner sep=0pt, minimum size=5pt, label=above:{$x_2$}](x2) at (1,1) {};
    \node[circle, fill, inner sep=0pt, minimum size=5pt, label=below:{$z_3$}](z3) at (2,0) {};
    \node[circle, fill, inner sep=0pt, minimum size=5pt, label=above:{$x_3$}](x3) at (2,1) {};
    \node[circle, fill, inner sep=0pt, minimum size=5pt, label=below:{$z_4$}](z4) at (3,0) {};
    \node[circle, fill, inner sep=0pt, minimum size=5pt, label=above:{$x_4$}](x4) at (3,1) {};
    \node[circle, fill, inner sep=0pt, minimum size=5pt, label=below:{$z_5$}](z5) at (4,0) {};
    \node[circle, fill, inner sep=0pt, minimum size=5pt, label=above:{$x_5$}](x5) at (4,1) {};
    \node[circle, fill, inner sep=0pt, minimum size=5pt, label=below:{$u$}](u) at (5,0) {};
    \draw (z1) -- (x1);
    \draw (z2) -- (x2);
    \draw (z3) -- (x3);
    \draw (z4) -- (x4);
    \draw (z5) -- (x5);
    \draw (z1) -- (z2);
    \draw (z2) -- (z3);
    \draw (z4) -- (z3);
    \draw (z4) -- (z5);
    \draw (x2) -- (z1);
    \draw (x3) -- (z2);
    \draw (x4) -- (z3);
    \draw (x5) -- (z4);
    \draw (u) -- (z5);
\end{tikzpicture}} \subcaptionbox{The graph $D_{\textsc{QParity}_5}^{\Dtrv}\!.$\label{fig:parity_dep}}[.5\textwidth]{\begin{tikzpicture}
\node[circle, fill, inner sep=0pt, minimum size=5pt, label=below:{$z_1$}](z1) at (0,0) {};
    \node[circle, fill, inner sep=0pt, minimum size=5pt, label=above:{$x_1$}](x1) at (0,1) {};
    \node[circle, fill, inner sep=0pt, minimum size=5pt, label=below:{$z_2$}](z2) at (1,0) {};
    \node[circle, fill, inner sep=0pt, minimum size=5pt, label=above:{$x_2$}](x2) at (1,1) {};
    \node[circle, fill, inner sep=0pt, minimum size=5pt, label=below:{$z_3$}](z3) at (2,0) {};
    \node[circle, fill, inner sep=0pt, minimum size=5pt, label=above:{$x_3$}](x3) at (2,1) {};
    \node[circle, fill, inner sep=0pt, minimum size=5pt, label=below:{$z_4$}](z4) at (3,0) {};
    \node[circle, fill, inner sep=0pt, minimum size=5pt, label=above:{$x_4$}](x4) at (3,1) {};
    \node[circle, fill, inner sep=0pt, minimum size=5pt, label=below:{$z_5$}](z5) at (4,0) {};
    \node[circle, fill, inner sep=0pt, minimum size=5pt, label=above:{$x_5$}](x5) at (4,1) {};
    \node[circle, fill, inner sep=0pt, minimum size=5pt, label=below:{$u$}](u) at (5,0) {};
    \draw (z1) -- (x1);
    \draw (z2) -- (x2);
    \draw (z3) -- (x3);
    \draw (z4) -- (x4);
    \draw (z5) -- (x5);
    \draw (z1) -- (z2);
    \draw (z2) -- (z3);
    \draw (z4) -- (z3);
    \draw (z4) -- (z5);
    \draw (x2) -- (z1);
    \draw (x3) -- (z2);
    \draw (x4) -- (z3);
    \draw (x5) -- (z4);
    \draw (u) -- (z5);
     \draw [-{Stealth}, red] (x1) -- (z2);
     \draw [-{Stealth}, red] (x1) -- (z3);
     \draw [-{Stealth}, red] (x1) -- (z4);
     \draw [-{Stealth}, red] (x1) -- (z5);
     \draw [-{Stealth}, red] (x1) -- (u);
     \draw [-{Stealth}, red] (x2) -- (z3);
     \draw [-{Stealth}, red] (x2) -- (z4);
     \draw [-{Stealth}, red] (x2) -- (z5);
     \draw [-{Stealth}, red] (x2) -- (u);
     \draw [-{Stealth}, red] (x3) -- (z1);
     \draw [-{Stealth}, red] (x3) -- (z4);
     \draw [-{Stealth}, red] (x3) -- (z5);
     \draw [-{Stealth}, red] (x3) -- (u);
     \draw [-{Stealth}, red] (x4) -- (z1);
     \draw [-{Stealth}, red] (x4) -- (z2);
     \draw [-{Stealth}, red] (x4) -- (z5);
     \draw [-{Stealth}, red] (x4) -- (u);
     \draw [-{Stealth}, red] (x5) -- (u);
     \draw [-{Stealth}, red] (x5) -- (z1);
     \draw [-{Stealth}, red] (x5) -- (z2);
     \draw [-{Stealth}, red] (x5) -- (z3);
     \draw [-{Stealth}, red, bend angle=45, bend left] (u) to (z1);
     \draw [-{Stealth}, red, bend angle=45, bend left] (u) to (z2);
     \draw [-{Stealth}, red, bend angle=45, bend left] (u) to (z3);
     \draw [-{Stealth}, red, bend angle=45, bend left] (u) to (z4);
     \draw [-{Stealth}, red, bend angle=45, bend left] (u) to (z5);
\end{tikzpicture}\vspace{-0.5cm}}
\subcaptionbox{A trunk-aligned tree decomposition of $\textsc{QParity}_n$ with respect to $\Dtrv$.\label{fig:parity_treedecomp}} {\begin{tikzpicture}
    \node[draw, inner sep=2.5pt, rounded corners](e1) at (-0.9,0) {$\emptyset$};
    \node[draw, inner sep=2.5pt, rounded corners](1) at (0.9-1,0) {$x_1$};
    \node[draw, inner sep=2.5pt, rounded corners](2) at (2.1-1.1,0) {$x_1, z_1$};
    \node[draw, inner sep=2.5pt, rounded corners](3) at (3.3-1.2,0) {$z_1$};
    \node[draw, inner sep=2.5pt, rounded corners](4) at (4.5-1.3,0) {$z_1,x_2$};
    \node[draw, inner sep=2.5pt, rounded corners](5) at (6.2-1.4,0) {$z_1,x_2,z_2$};
    \node[draw, inner sep=2.5pt, rounded corners](6) at (7.8-1.5,0) {$x_2,z_2$};
    \node[draw, inner sep=2.5pt, rounded corners](7) at (8.9-1.6,0) {$z_2$};
    \node[inner sep=2.5pt](8) at (9.7-1.7,0) {$\dots$};
    \node[draw, inner sep=2.5pt, rounded corners](9) at (10.9-1.7,0) {$x_n,z_n,u$};
    \node[draw, inner sep=2.5pt, rounded corners](10) at (12.3-1.7,0) {$z_n,u$};
    \node[draw, inner sep=2.5pt, rounded corners](11) at (13.2-1.7,0) {$u$};
    \node[draw, inner sep=2.5pt, rounded corners](e2) at (12.2,0) {$\emptyset$};
    \node[] at (0,0.8) {};
    \node[] at (0,-0.2) {};
    \draw (e1) -- (1) -- (2) -- (3) -- (4) -- (5) -- (6) -- (7) -- (8) -- (9) -- (10) -- (11) -- (e2);
\end{tikzpicture}}
\caption{Parity-related graphs and decompositions 
that are used in Lemmas~\ref{lem:dtw_qpar},~\ref{lem:ppw_parity}, and \ref{lem:btw_parity}.}
\end{figure}
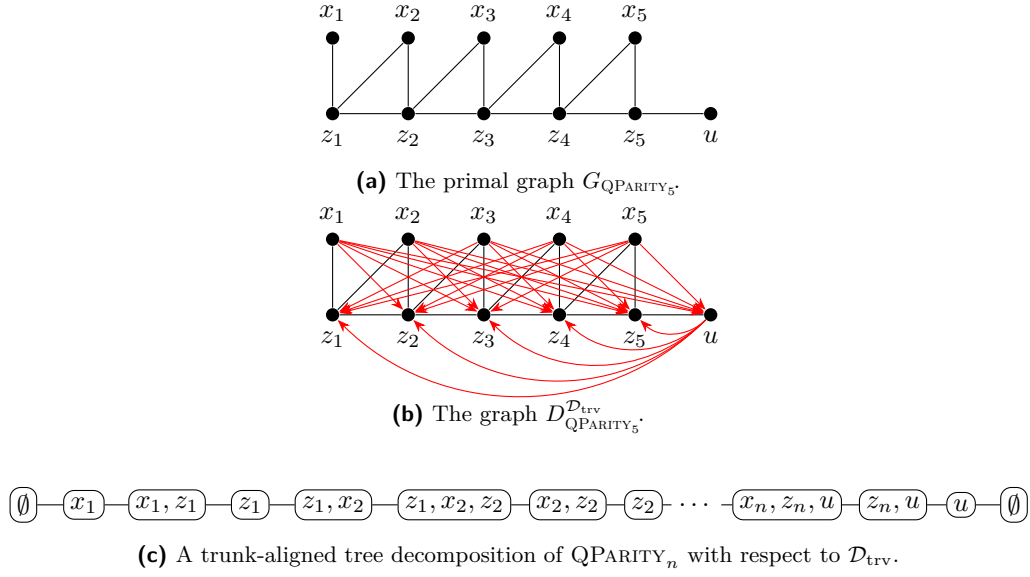

     \begin{lemma}
\label{lem:dtw_qpar}
    For each $n \in \mathbb{N}$ with $n\geq 2$, we have $\dtw{\emph{\qpar}, \Dtrv}\geq n+1$.
\end{lemma}
\begin{proof}
Note that for each $i\in [n]$, we have $x_i \preceq_{\D_{\mathrm{trv}}}u$ and $u \preceq_{\D_{\mathrm{trv}}} z_i$. 
Since a $\Dtrv$-elimination ordering must be a linear extension of the reverse of $\D_{\mathrm{trv}}$, we know every such ordering $\mathcal{V}$ must consist of a sequence of the variables in $\mathcal{Z}_n$, followed by $u$, followed by a sequence of the variables in $\mathcal{X}_n$. Since $\mathcal{Z}_n$ is connected and its neighborhood in $G_{\qpar}$ is $\mathcal{X}_n\cup \{u\}$, the last variable eliminated from $\mathcal{Z}_n$ must be a neighbor of every variable of $\mathcal{X}_n\cup \{u\}$ in the fill-in graph $G_{\qpar}^\mathcal{V}\!.$ To see this, note that by the definition of $G_{\qpar}^\mathcal{V}$, after each elimination of a variable in $\mathcal{Z}_n$, all of its current neighbors are transferred to its not-yet-eliminated neighbors in $\mathcal{Z}_n$.

In particular, the width of $\mathcal{V}$ is at least $n+1$, and hence $\dtw{\qpar, \Dtrv}\geq n+1$.
\end{proof}

   \begin{lemma}\label{lem:ppw_parity}
    For each $n \in \mathbb{N}$ with $n \geq 2$, we have $\ppw{\emph{\qpar}, \Dtrv}\geq n$.
\end{lemma}
\begin{proof}
    Due the characterization of prefix pathwidth via directed pathwidth, it suffices to show $\dpw{\qpar,\Dtrv}\geq n$ for each $n \in \mathbb{N}$ with $n \geq 2$. An illustration of $D_{\qpar}^{\Dtrv}$ for $n=5$ is provided in Figure~\ref{fig:parity_dep}, where undirected edges indicate two anti-parallel directed edges. Recalling the cops-and-robber characterization of directed pathwidth that was outlined in Section~\ref{sec:prelims}, we prove the statement by showing that $\cn{D_{\qpar}^{\Dtrv}}>n$, that is, $n$ cops do not suffice to catch the robber on~$D_{\qpar}^{\Dtrv}\!.$

        Consider the following strategy for the robber: Upon the initialization of the game, the robber selects an arbitrary vertex in $ \mathcal{X}_n$; let us call this vertex $x_{\mathrm{robber}}$. Let the cops choose $n$ vertices that they will occupy. Now the robber plays as follows: If none of the cops occupies the vertex $x_{\mathrm{robber}}$, then the robber simply stays there. If a cop moves to the vertex $x_{\mathrm{robber}}$, then, either (a) there exists some $i \in [n]$ such that neither $x_i$ nor $z_i$ is occupied by any cop or (b) by the pigeon-hole principle, the vertex $u$ is not occupied by any cop.
    By construction of $D_{\qpar}^{\Dtrv}$, there is a directed edge from $x_{\mathrm{robber}}$ to $z_i$, from $z_i$ to $x_i$, and from $x_{\mathrm{robber}}$ to $u$. Consequently, in Case (a), the robber can move from $x_{\mathrm{robber}}$ to $x_i$, which now becomes the new $x_{\mathrm{robber}}$. In Case (b), the robber moves to $u$ and stays there until $u$ gets occupied by some cop. Then, analogously to Case (a), the robber can transition to a non-occupied $x_i$ via a non-occupied $z_i$, since $(u,z_i), (z_i, x_i) \in E(D_{\qpar}^{\Dtrv})$.
    A robber that plays according to this strategy 
     stays on a vertex $ v \in\mathcal{X}_n \cup \{u\}$ at all times and is never caught. 
\end{proof}
Having shown that on $\qpar$, both the dependency treewidth and the prefix pathwidth with respect to $\Dtrv$ are high, we conclude the proof of Theorem~\ref{thm:parity} by showing that the respective bilateral treewidth is a small constant.
   \begin{lemma}\label{lem:btw_parity}
    For each $n \in \mathbb{N}$ with $n \geq 2$, we have $\btw{\emph{\qpar}, \Dtrv}=2$.
\end{lemma}
\begin{proof}
    Figure~\ref{fig:parity_treedecomp} shows a trunk-aligned tree decomposition $\tree$ of $\qpar$ with respect to~$\Dtrv$ of width $2$. The decomposition consists of a single path, the trunk $\trunk$, that is ordered from left to right. Let us now verify that $\tree$ is indeed a trunk-aligned tree decomposition. It is not hard to see that the Properties~\ref{itm:T1} to \ref{itm:T4} of regular tree decompositions are all satisfied by $\tree$. For each variable $v \in \mathcal{X}_n $, we have that all variables that $v$ depends on are forgotten ``below'' the node in which $v$ is forgotten---hence all variables in $\mathcal{X}_n$ satisfy \ref{itm:P2}. For each variable $v \in \mathcal{Z}_n $, we have that all variables that depend on $v$ avoid the node in which $v$ is forgotten---hence all variables in $\mathcal{Z}_n$ satisfy \ref{itm:P1}. Finally, $u$ satisfies both \ref{itm:P1} and \ref{itm:P2}, since all other variables are forgotten ``below'' the node in which $u$ is forgotten.
     \end{proof}

 \section{Using Trunk-Aligned Tree Decompositions to Solve QBF}\label{sec:fpt_proof}
 In this section, we show how the decompositions underlying bilateral treewidth yield fixed-parameter tractability for $\qsat$. More precisely, we establish:
 
\begin{theorem}\label{thm:fpt_btw}
There is a computable function~$f$ and an algorithm that given
    \begin{enumerate}
        \item a QBF $\Phi = \Q. \phi$,
        \item a dependency poset $\D$, and
        \item a trunk-aligned tree decomposition $\tree=(\mathbf{T}, \rchi)$ of $(\Phi, \D)$ with width $\btw{\Phi, \D}$
    \end{enumerate}
    decides whether $\Phi$ is true in time $f(\btw{\Phi, \D})\cdot |\Phi|\cdot |\tree|$.
\end{theorem}

Below, we establish several smaller lemmas that will set up the proof of Theorem~\ref{thm:fpt_btw}. 
Broadly speaking, we will follow the following proof strategy. Based on a given trunk-aligned tree decomposition~$\tree$, we define a proof-like structure that maintains sets of sets of QBFs and that we will refer to as the \emph{derivation sequence}~$\Pi$. We then prove that $\Pi$ satisfies specific properties, namely: the variable-co-occurrences in QBFs within $\Pi$ can be tightly linked to the content of bags in $\tree$ (Lemma~\ref{lem:bounded_neighborhood}) and $\Pi$ preserves the truth values of the QBFs it maintains (Lemma~\ref{lem:main_invariant}), while gradually eliminating all variables from
\iflong
them (Observation~\ref{obs:total_elim}).
\fi
\ifshort
them.
\fi
Moreover, the number of QBFs that is maintained in parallel at each time step is bounded by a function in the bilateral treewidth of $\Phi$ (Lemma~\ref{lem:equiv_formulas}) from which we will infer that $\Pi$ can be efficiently computed (Lemma~\ref{lem:rules_efficient}). The proof of Theorem~\ref{thm:fpt_btw} then concludes this section.
       In order to apply Lemmas~\ref{lem:res_correctness},~\ref{lem:red_correctness} and~\ref{lem:ext_correctness},
we will assume that the given QBF~$\qbf$ does not contain any tautological clauses; this is without loss of generality as such clauses can be removed in a pre-processing step without altering the truth value of~$\Phi$.

As a first step towards obtaining a derivation sequence, our proposed algorithm takes the partial order $\lesst$ imposed by the trunk-aligned tree decomposition $\tree$ and extends it to a full linear order~$<_{\mathbf{T}'}$, that is, for each two distinct $t,t' \in V(\mathbf{T})$ we have $t <_{\mathbf{T}'} t'$ or $t' <_{\mathbf{T}'} t$, and if $t <_{\mathbf{T}} t'$, then also $t <_{\mathbf{T}'} t'$. Now, let $\V=(v_1, \dots, v_n)$ be the sequence of variables in $\var{\Phi}$ such that $\forget{v_1} \lesstp \forget{v_2} \lesstp \cdots \lesstp \forget{v_n}$. We call $\V$ an \emph{elimination ordering of $\var{\phi}$ that agrees with $\tree$}.

Based on the elimination ordering $\V=(v_1, \dots, v_n)$, our algorithm constructs a \emph{derivation sequence} $\Pi = \Pi_0, \Pi_1, \dots, \Pi_n$ of $\qbf$, where each $\Pi_i=(\Q^{(i)}, \mathcal{F}^{(i)})$ is a tuple that consists of a prefix $\Q^{(i)}$ and a set $\mathcal{F}^{(i)}$ of sets of matrices as follows. We initialize $\Q^{(0)}=\Q$ and $\mathcal{F}^{(0)}=\{\{\phi\}\}$. For each $i \in [n]$, we compute $\Q^{(i)}$ and $\mathcal{F}^{(i)}$ according to the following rules:
    \begin{description}[labelwidth=22pt,leftmargin=33pt]
        \item[\namedlabel{itm:R1}{(R1)}] If $v_i \notin \var{\Q^{(i-1)}}$, then set $\Q^{(i)}=\Q^{(i-1)}$ and $\mathcal{F}^{(i)}=\mathcal{F}^{(i-1)}$.          \item[\namedlabel{itm:R2}{(R2)}] If $v_i \in \vare{\Q^{(i-1)}}$ and $\{v \in \var{\Q^{(i-1)}} \mid v_i \dep v\} \cap \rchi(\forget{v_i})=\emptyset$, then for each set $\pi \in \mathcal{F}^{(i-1)}$ of matrices, we construct $\pi' \in \mathcal{F}^{(i)}$ such that $\pi'$ contains the matrix $\res{\psi}{v_i}$ if and only if $\pi$ contains the matrix $\psi$. Moreover, we set $\Q^{(i)}=\Q^{(i-1)}\setminus \{v_i\}$.        
        \item[\namedlabel{itm:R3}{(R3)}] If $v_i \in \vara{\Q^{(i-1)}}$ and $\{v \in \var{\Q^{(i-1)}} \mid v_i \dep v\} \cap \rchi(\forget{v_i})=\emptyset$, then for each set $\pi \in \mathcal{F}^{(i-1)}$ of matrices, we construct $\pi' \in \mathcal{F}^{(i)}$ such that $\pi'$ contains the matrix $\red{\psi}{v_i}$ if and only if $\pi$ contains the matrix $\psi$. We set $\Q^{(i)}=\Q^{(i-1)}\setminus \{v_i\}$.
        \item[\namedlabel{itm:R4}{(R4)}] Otherwise, we let $\mathcal{F}^{(i)}=\bigcup_{\pi \in \mathcal{F}^{(i-1)}} \strext{\pi}{v_i}$ and $\Q^{(i)}=\Q^{(i-1)}\setminus \depless{v_i}$.
    \end{description}
Intuitively, we alter sets of matrices by gradually eliminating all variables in the order specified by $\V$. Variables are eliminated either by applying the resolution rule to each QBF (Rule~\ref{itm:R2}), by applying the reduction rule to each QBF (Rule~\ref{itm:R3}), or by performing strategy extension on each current set of matrices (Rule~\ref{itm:R4}). Rule~\ref{itm:R1} merely deals with a case that may arise from the earlier application of strategy extension.

To prove the correctness of the Rules \ref{itm:R2} and \ref{itm:R3}, we crucially rely on the fact that whenever a variable~$v_i$ is eliminated by this rule, all variables that it is neighbored to in any matrix $\psi \in \pi$ with $\pi \in \mathcal{F}^{(i-1)}$ are contained in $\rchi(\forget{v_i})$. We prove this statement in the following.

\begin{lemma}
\label{lem:bounded_neighborhood}
    Let $\qbf$ be a QBF, $\tree=(\mathbf{T}, \rchi)$ be a trunk-aligned tree decomposition of $(\Phi, \D)$ with width $k=\btw{\Phi, \D}$ and $\V$ be an elimination ordering that agrees with $\tree$. Further, let $\Pi=\Pi_0, \Pi_1, \dots, \Pi_n$ be a derivation sequence of $\Phi$ obtained by the rules above.
    Then for each $i \in [n]$, we have that $N_{G_{\psi}}(v_i) \subseteq \rchi(\forget{v_{i}})$ for every matrix $\psi \in \pi$ in every set $\pi \in \mathcal{F}^{(i-1)}$, where $N_{G_{\psi}}(v_i)$ is the open neighborhood of $v_i$ in the primal graph of $\psi$.
     \end{lemma}

\iflong
\begin{proof}          We prove the lemma by induction on $i$.  
    For $i=1$, we have $\mathcal{F}^{(0)}=\{\{\phi\}\}$. From the definitions of $\tree$ and $\V$ and the fact that $v_1$ is eliminated first in~$\V$, we conclude $N_{G_{\phi}}(v_1) \subseteq \rchi(\forget{v_1})$. 

    For the inductive step, let $\psi$ be an arbitrary matrix that occurs in some $\pi \in \mathcal{F}^{(i-1)}$ for some $i \in [n]$. Towards a contradiction, suppose there is some $v' \in N_{G_{\psi}}(v_i)$ with $v' \notin \rchi(\forget{v_{i}})$. Note that $v'=v_j$ for some $j > i$, since $v'$ is still present in $\psi$ when $v_i$ is forgotten. Consequently, if there was a clause in the original matrix $\phi$ that contained both $v'$ and $v_{i}$, then $v' \in \rchi(\forget{v_{i}})$, contradicting our assumption. Let us thus assume that there was no such clause. Observe that applying the Rules~\ref{itm:R1}, \ref{itm:R3}, and \ref{itm:R4} does not increase the neighborhood of variables in any matrix; hence, there must have been an application of Rule~\ref{itm:R2} that introduced a clause that contains both $v'$ and $v_i$.
    Formally, there exists a $v_{j'}\in \vare{\Phi}$ with $1 <j' <i$, such that some $\psi'$ in some $\pi' \in\mathcal{F}^{(j'-1)}$ contained a clause $\clause_1$ with $\{v_{j'},v'\} \subseteq\var{\clause_1}$ and a clause $\clause_2$ with $\{v_{j'},v_i\} \subseteq\var{\clause_2}$ and, without loss of generality, $v_{j'} \in \clause_1$ and $\overline{v_{j'}} \in \clause_2$. Let us consider every possible relative position of $\forget{v_{j'}}$ and $\forget{v_{i}}$ in $\tree$.
    \begin{description}
        \item[Case 1]  It is impossible that $\forget{v_{i}}\lesst \forget{v_{j'}}$, since $j' < i$ and $\V$ agrees with $\tree$.
        \item[Case 2]  Suppose $\forget{v_{j'}}\lesst \forget{v_{i}}$, that is, both nodes lie on the same path in~$\tree$. By the induction hypothesis, we have that $\{v_{i}, v'\} \subseteq \rchi(\forget{v_{j'}})$. However, since $v'=v_j \in \rchi(\forget{v_j})$ for some $j >i$ and $\forget{v_{j'}}\lesst \forget{v_{i}}$, we have $ v'\in \rchi(\forget{v_i})$ by the definition of a tree decomposition, contradicting our assumption.
        \item[Case 3] Suppose neither $\forget{v_{i}}\lesst \forget{v_{j'}}$ nor $\forget{v_{j'}}\lesst \forget{v_{i}}$, that is, these nodes lie on different paths in $\tree$. By the induction hypothesis, we have that $\{v_{i}, v'\} \subseteq \rchi(\forget{v_{j'}})$. Moreover, applying the definition of $\forget{v_i}$, we have that $v_{i} \in \rchi(\forget{v_{i}})$ but $v_{i} \notin \rchi(t)$ for any node $t$ on the path from $\forget{v_i}$ to $\rt$. But then, the nodes whose bags contain $v_i$ are not connected in $\tree$, contradicting the definition of tree decompositions.   \qedhere
    \end{description}
\unskip\end{proof}
\fi
Note that it follows from Lemma~\ref{lem:bounded_neighborhood} that for all $i \in [n]$, $|N_{G_{\psi}}(v_i)|\leq k$ for all $\psi \in \pi$  and $\pi \in \mathcal{F}^{(i-1)}$, since  $|\rchi(\forget{v_{i}})|\leq k+1$ and $v_i \in \rchi(\forget{v_{i}})$ by definition.
We will use this fact in Lemma~\ref{lem:rules_efficient} to prove that the Rules~\ref{itm:R2} and~\ref{itm:R3} can be applied efficiently.
For now, let us prove that the constructed derivation sequence preserves the truth value of the original QBF by maintaining the following invariant.
\begin{lemma}
\label{lem:main_invariant}
    Let $\Pi = \Pi_0, \Pi_1, \dots, \Pi_n$ be a derivation sequence of $\qbf$.
     Then for all $i \in [n]$, there exists a $\pi' \in \mathcal{F}^{(i)}$ such that the set of QBFs $\{\Q^{(i)}. \psi' \mid \psi' \in \pi'\}$ is true
      if and only if there exists a $\pi \in \mathcal{F}^{(i-1)}$ such that the set of QBFs $\{\Q^{(i-1)}.\psi \mid \psi \in \pi\}$
      is true.
\end{lemma}
\ifshort
\begin{proof}[Proof Sketch]
Recall that each $\Pi_i$ is obtained from its predecessor $\Pi_{i-1}$ by applying one of the Rules~\ref{itm:R1} to~\ref{itm:R4}. Their correctness can be essentially inferred from the correctness of resolution (Lemma~\ref{lem:res_correctness}), reduction (Lemma~\ref{lem:red_correctness}), and strategy extension (Lemma~\ref{lem:ext_correctness}). For the Rules~\ref{itm:R2} and~\ref{itm:R3}, we additionally make use of Lemma~\ref{lem:bounded_neighborhood} to ensure that
$\rchi(\forget{v_i})$ contains in particular all variables that depend on $v_i$ and that share a clause with $v_i$ in some $\psi \in \pi$ with $\pi \in \mathcal{F}^{(i-1)}$.
\end{proof}
\fi
\iflong
\begin{proof}     Let $\Pi = \Pi_0, \Pi_1, \dots, \Pi_n$ be a derivation sequence of $\Phi$ in which each $\Pi_i=(\Q^{(i)}, \mathcal{F}^{(i)})$ was obtained from its predecessor $\Pi_{i-1}=(\Q^{(i-1)}, \mathcal{F}^{(i-1)})$ by applying one of the Rules~\ref{itm:R1} to~\ref{itm:R4}. We show that each rule preserves the desired invariant. Towards this goal, we show that if no matrix that occurs in some set in $\mathcal{F}^{(i-1)}$ contains a tautological clause, then this is also the case for $\mathcal{F}^{(i)}$. Note that $\mathcal{F}^{(0)}=\{\{\phi\}\}$ satisfies this property by assumption.
    \begin{description}         \item[Case 1] Suppose $\Pi_i$ was obtained by applying~\ref{itm:R1}. Then $\Pi_i=\Pi_{i-1}$ which clearly preserves the invariant and does not introduce any tautologies.
        \item[Case 2] Suppose $\Pi_i$ was obtained by applying~\ref{itm:R2}. By induction, no set in $\mathcal{F}^{(i-1)}$ contains a matrix with a tautological clause. Moreover, by Lemma~\ref{lem:bounded_neighborhood}, $\rchi(\forget{v_i})$ contains all variables that $v_i$ is neighbored to in any $\psi \in \pi$ with $\pi \in \mathcal{F}^{(i-1)}$, but since Rule~\ref{itm:R2} was applied, none of these variables depends on $v_i$.
        It follows from Lemma~\ref{lem:res_correctness} that the invariant is preserved and that no set in $\mathcal{F}^{(i)}$ contains a matrix with a tautological clause.
        \item[Case 3] Suppose $\Pi_i$ was obtained by applying~\ref{itm:R3}. By induction, no set in $\mathcal{F}^{(i-1)}$ contains a matrix with a tautological clause. Moreover, by Lemma~\ref{lem:bounded_neighborhood}, $\rchi(\forget{v_i})$ contains all variables that $v_i$ is neighbored to in any $\psi \in \pi$ with $\pi \in \mathcal{F}^{(i-1)}$, but since Rule~\ref{itm:R3} was applied, none of these variables depends on $v_i$.
        It follows from Lemma~\ref{lem:red_correctness} that the invariant is preserved and that no set in $\mathcal{F}^{(i)}$ contains a matrix with a tautological clause.
        \item[Case 4] Suppose $\Pi_i$ was obtained by applying~\ref{itm:R4}. By Lemma~\ref{lem:ext_correctness}, there is a $\pi \in \mathcal{F}^{(i-1)}$ such that all QBFs in $\{\Q^{(i-1)}.\psi\mid \psi \in \pi\}$ 
        are true if and only if there is a $\pi' \in \strext{\pi}{v_i}$ such that all QBFs in $\{\Q^{(i)}.\psi' \mid \psi' \in \pi'\}$
        are true. By the definition of Rule~\ref{itm:R4}, we have that $\pi' \in \mathcal{F}^{(i)}$ if and only if $\pi \in \mathcal{F}^{(i-1)}$, which concludes the proof of the invariant. By induction, no  $\pi' \in\mathcal{F}^{(i)}$ contains a matrix with a tautological clause. \qedhere
    \end{description}
\unskip\end{proof}
\fi
 Since $\Pi_n$ will be ultimately used to decide the truth value of the original QBF $\Phi$, it is crucial that this decision can be made efficiently on all matrices in sets of $\Pi_n$. 
This is indeed the case as they will not contain any variables, and will hence either be empty or contain the empty clause. 
 \iflong
We prove this in the following.

\begin{longobservation}\label{obs:total_elim}
    Let $\Pi=\Pi_0, \Pi_1, \dots, \Pi_n$ be a derivation sequence of $\qbf$ and $\V=(v_1, \dots, v_n)$ the elimination ordering that was used to compute it. Then for each $i \in [n]$, there is no matrix $\psi$ in any $\pi \in \mathcal{F}^{(i)}$ that contains any variable $v_j$ with $1 \leq j \leq i$.
\end{longobservation}
\begin{proof}
    One can perform an induction on $i$, using the latter statements of Lemma~\ref{lem:res_correctness},~\ref{lem:red_correctness}, and~\ref{lem:ext_correctness} that ensure that variables are eliminated exhaustively whenever resolution, reduction, or strategy extension is performed. Note that whenever $v_i \notin \var{\Q^{(i-1)}}$, that is, Rule~\ref{itm:R1} is used to derive $\mathcal{F}^{(i)}$, we can conclude from $\var{\phi}\subseteq \var{\Q}$ that there must have been a variable $v_{i'}$ with $1\leq i' < i$ and $v_{i} \in \depless{v_{i'}}$ that was eliminated before $v_i$, using Rule~\ref{itm:R4}. Then, by Lemma~\ref{lem:ext_correctness}, $v_i$ was eliminated from all matrices as well.
\end{proof}
\fi
Before we proceed to analyze the complexity of computing derivation sequences, we illustrate their construction on a small example.
\begin{example}
We want to compute the derivation sequence for $\textsc{QParity}_2$ with respect to the trivial dependency poset $\Dtrv$ and the following trunk-aligned tree decomposition~$\tree$.

\begin{figure}[h!]
\centering
\begin{tikzpicture}
    \node[draw, inner sep=2.5pt, rounded corners](e1) at (-0.9,0) {$\emptyset$};
    \node[draw, inner sep=2.5pt, rounded corners](1) at (0.9-1,0) {$x_1$};
    \node[draw, inner sep=2.5pt, rounded corners](2) at (2.1-1.1,0) {$x_1, z_1$};
    \node[draw, inner sep=2.5pt, rounded corners](3) at (3.3-1.2,0) {$z_1$};
    \node[draw, inner sep=2.5pt, rounded corners](4) at (4.5-1.3,0) {$z_1,x_2$};
    \node[draw, inner sep=2.5pt, rounded corners](5) at (6.2-1.4,0) {$z_1,x_2,z_2$};
    \node[draw, inner sep=2.5pt, rounded corners](6) at (7.9-1.5,0) {$x_2,z_2$};
    \node[draw, inner sep=2.5pt, rounded corners](7) at (9.5-1.6,0) {$x_2,z_2,u$};
    \node[draw, inner sep=2.5pt, rounded corners](8) at (10.9-1.6,0) {$z_2,u$};
    \node[draw, inner sep=2.5pt, rounded corners](9) at (11.9-1.7,0) {$u$};
    \node[draw, inner sep=2.5pt, rounded corners](e2) at (12.7-1.7,0) {$\emptyset$};
    \draw (e1) -- (1) -- (2) -- (3) -- (4) -- (5) -- (6) -- (7) -- (8) -- (9) -- (e2);
\end{tikzpicture}
\end{figure}
 
Since $\tree$ is a path, the only elimination ordering of $\var{\textsc{QParity}_2}$ that agrees with $\tree$ is $\V=(x_1, z_1,x_2, z_2, u)$. We initialize $\Q^{(0)}=\exists x_1 x_2 \forall u \exists z_1 z_2$ and $\mathcal{F}^{(0)}=\{\{\phi\}\}$ with $\phi=(x_1 \lor \overline{z_1}) \land (\overline{x_1} \lor z_1)  \land (u \lor \overline{z_2}) \land (\overline{u} \lor z_2)  \land  (\overline{z_{2}} \lor x_{2} \lor z_1) \land (z_{2} \lor \overline{x_{2}} \lor z_1) \land (z_{2} \lor x_{2} \lor \overline{z_1}) \land (\overline{z_{2}} \lor \overline{x_{2}} \lor \overline{z_1})$.

Now observe that $x_1$, the first variable that is eliminated in $\V$, is existential but does not satisfy Rule~\ref{itm:R2} since $z_1 \in \rchi(\forget{x_1})$ and $x_1\preceq_{\D_{\mathrm{trv}}}z_1$. Therefore, Rule~\ref{itm:R4} is applied, which means that $\mathcal{F}^{(1)}$ is computed as $\strext{\{\phi\}}{x_1}$ via strategy extension. Since $x_1$ does not depend on any universal variable, the only two strategies for $x_1$ are setting it to constant 0 or 1. Therefore, we have $\strext{\{\phi\}}{x_1}=\{\{\psi_{x_1=0}\}, \{\psi_{x_1=1}\}\}$ with $\psi_{x_1=0}=(\overline{z_1}) \land (u \lor \overline{z_2}) \land (\overline{u} \lor z_2) \land \mathrm{xor}(z_2, x_2, z_1)$ and $ \psi_{x_1=1}= (z_1)  \land (u \lor \overline{z_2}) \land (\overline{u} \lor z_2) \land \mathrm{xor}(z_2, x_2, z_1)$, where $\mathrm{xor}(z_2, x_2, z_1)=(\overline{z_{2}} \lor x_{2} \lor z_1) \land (z_{2} \lor \overline{x_{2}} \lor z_1) \land (z_{2} \lor x_{2} \lor \overline{z_1}) \land (\overline{z_{2}} \lor \overline{x_{2}} \lor \overline{z_1})$.
The updated prefix is $\Q^{(1)}=\exists x_2 \forall u \exists z_1 z_2$.

Next, $z_1$ is eliminated. Observe that the node in which $z_1$ is forgotten does not contain any variable that depends on $z_1$. Hence, we apply \ref{itm:R2}, that is, we separately resolve over $z_1$ in all matrices. As a result, we obtain $\mathcal{F}^{(2)}=\{\pi_1=\{\res{\psi_{x_1=0}}{z_1}\}, \pi_2=\{\res{\psi_{x_1=1}}{z_1}\}\}$ with $\res{\psi_{x_1=0}}{z_1}=(u \lor \overline{z_2}) \land (\overline{u} \lor z_2)  \land  (\overline{z_{2}} \lor x_{2}) \land (z_{2} \lor \overline{x_{2}})$ and $\res{\psi_{x_1=1}}{z_1}=(u \lor \overline{z_2}) \land (\overline{u} \lor z_2)  \land (z_{2} \lor x_{2} ) \land (\overline{z_{2}} \lor \overline{x_{2}})$. The new prefix is $\Q^{(2)}=\exists x_2 \forall u \exists z_2$.

Now we eliminate $x_2$ via \ref{itm:R4}. This gives rise to $\mathcal{F}^{(3)}=\strext{\pi_1}{x_2}\cup \strext{\pi_2}{x_2}$. Note that again, the only possible strategies for $x_2$ are setting it to constant 0 or 1, and we account for both with respect to each set of $\mathcal{F}^{(2)}$. More specifically, we obtain 
 $\strext{\pi_1}{x_2}=\{\{\psi^{\pi_1}_{x_2=0}, \psi^{\pi_1}_{x_2=1}\}\}$ and $\strext{\psi_2}{x_2}=\{\{\psi^{\pi_2}_{x_2=0}, \psi^{\pi_2}_{x_2=1}\}\}$, where $\psi^{\pi_1}_{x_2=0}=(u \lor \overline{z_2}) \land (\overline{u} \lor z_2)  \land  \overline{z_{2}}$, $\psi^{\pi_1}_{x_2=1}=(u \lor \overline{z_2}) \land (\overline{u} \lor z_2)  \land z_{2} $, $\psi^{\pi_2}_{x_2=0}=(u \lor \overline{z_2}) \land (\overline{u} \lor z_2)  \land z_{2} $, and $\psi^{\pi_2}_{x_2=1}=(u \lor \overline{z_2}) \land (\overline{u} \lor z_2)  \land \overline{z_{2}}$. In particular, this would amount to $\mathcal{F}^{(3)}=\{\{\psi^{\pi_1}_{x_2=0}, \psi^{\pi_1}_{x_2=1}\} , \{\psi^{\pi_2}_{x_2=0}, \psi^{\pi_2}_{x_2=1}\}\}$. However, since $\psi^{\pi_1}_{x_2=0}=\psi^{\pi_2}_{x_2=1}$ and $\psi^{\pi_1}_{ x_2=1}=\psi^{\pi_2}_{x_2=0}$, the set $\mathcal{F}^{(3)}$ simplifies to $\{\{\psi^{\pi_1}_{x_2=0}, \psi^{\pi_1}_{x_2=1}\}\}$. The new matrix is $\Q^{(3)}=\forall u \exists z_2$.

As a penultimate step, we eliminate $z_2$ by Rule~\ref{itm:R2}, which creates $\mathcal{F}^{(4)}=\{\{\overline{u}, u\}\}$ and $\Q^{(4)}=\forall u$.
Finally, $u$ is eliminated by Rule~\ref{itm:R3}, which performs reduction separately on the remaining matrices $\overline{u}$ and $u$. This results in $\mathcal{F}^{(5)}=\{\{\bot, \bot\}\}$ and an empty prefix. Since each set in $\mathcal{F}^{(5)}$ contains only matrices that in turn contain the empty clause, the derivation sequence witnesses the fact that $\textsc{QParity}_2$ is false.
\end{example}

The next two lemmas demonstrate that computing $\Pi$ is fixed-parameter tractable by the bilateral treewidth of $\Phi$ when given all the ingredients that are listed in Theorem~\ref{thm:fpt_btw}. 

First, note that the efficiency of our rules crucially depends on the number of formulas that need to be considered in parallel. In the following, we prove that this number is bounded by a function of the bilateral treewidth of the instance.
\begin{lemma}
\label{lem:equiv_formulas}
    Let $\Pi=\Pi_0, \Pi_1,  \dots, \Pi_n$ be a derivation sequence of $\qbf$ obtained from an elimination ordering $\V=(v_1, \dots, v_n)$ that agrees with a given trunk-aligned tree decomposition $\tree=(\mathbf{T}, \rchi)$ of $(\Phi, \D)$ with width $k=\btw{\Phi, \D}$. Then for each $i \in \{0, \dots, n\}$, we have $|\mathcal{F}^{(i)}|\leq 2^{2^{2^{\oh{k}}}}$ and for each $\pi \in \mathcal{F}^{(i)}$ we have $|\pi|\leq 2^{2^{\oh{k}}}$.
\end{lemma}
\begin{proof}
    By definition, $\mathcal{F}^{(0)}=\{\{\phi\}\}$. Now, let $i \in [n]$. If we applied Rule~\ref{itm:R1} to compute~$\Pi_i$ from $\Pi_{i-1}$, then the induction hypothesis immediately yields the bound.    
    If we applied either Rule~\ref{itm:R2} or~\ref{itm:R3} to compute $\Pi_i$ from $\Pi_{i-1}$, then we either resolve or reduce $v_i$ in all matrices that occur in some $\pi \in \mathcal{F}^{(i-1)}$. Crucially, these operations are carried out \emph{within} each matrix and thus their number cannot increase. Thus, by induction, the claim holds.

    For the remainder of this proof, let us assume that we applied Rule~\ref{itm:R4} to compute~$\Pi_i$ from $\Pi_{i-1}$.
     Let $\pi$ and $ \pi' $ be two (not necessarily distinct) sets of matrices in $\mathcal{F}^{(i)}$ and let $\psi\in \pi$ and $\psi' \in \pi'$ be some arbitrarily selected matrices. In the following, we identify sufficient conditions for $\psi \equiv \psi'$, which we will then use to bound the sizes of $\pi$, $\pi'$, and $\mathcal{F}^{(i)}$. 
    
    Let $F$ be the set of all variables $v_j$ with $j \in [i]$ that have been forgotten via Rule~\ref{itm:R4} so far. Then, by the definition of strategy extension, each variable in $V=\bigcup_{v_j \in F} \var{\Q^{(j-1)}} \cap \depless{v_j}$ was assigned to either $0$ or $1$. Let $\alpha_1\in \langle V \rangle $ ($\alpha_2 \in \langle V \rangle $, respectively) be the assignment that was applied to $\phi$ to---together with the interleaved resolution and reduction steps with respect to all $v_j \notin F$, $j \in [i]$---obtain $\psi$ ($\psi'$, resp.). Furthermore, let $U$ be the set of all possible clauses (including the empty clause) over the variables in $\rchi(p) \setminus V$, where $p$ is the unique parent of the node $\forget{v_i}$ in $\mathbf{T}$. We define $U_{\psi}$ ($U_{\psi'}$, respectively) to be the set of clauses that occur in both $U$ and $\psi$ (in both $U$ and $\psi'$, resp.) and let $\alpha_1'= \alpha_1 {\upharpoonright}_{V \cap \rchi(p)}$ and  $\alpha_2'= \alpha_2 {\upharpoonright}_{V \cap \rchi(p)}$. We claim that these two measures suffice to determine whether $\psi$ and $\psi'$ are equivalent, which is formally stated in Claim~\ref{cl:equiv}. 
    
    To establish this equivalence, a pairwise comparison of the clauses in $\psi$ and $\psi'$, respectively, is needed. This comparison takes the ``history'' of a clause $\clause$ within the derivation sequence~$\Pi$ into account, that is, the sequence of clauses that were used in $\Pi$ to derive $\clause$ from the original clauses of the given matrix $\phi$. 
                    \ifshort
     We will refer to this sequence as the \emph{certifying sequence} $D=(C_1, \dots, C_t=\clause)$ of $\clause$ in $\Pi$.
           We assume a certifying sequence $D$ to be minimal in the sense that each clause $C_j$ with $j \in [t-1]$ is used to derive a subsequent clause in $D$, either by resolution, reduction, or partial assignment.
     \fi
    \iflong
    We formalize this notion in the following.
    \begin{longdefinition}
        Let $i \in [n]$ and let $\clause$ be a clause that occurs in some matrix $\psi \in \pi$ for some $\pi \in \mathcal{F}^{(i)}$.
        A \emph{certifying sequence $D=(C_1, \dots, C_t)$ of $\clause$}
         produced by the partial derivation sequence $\Pi_0, \Pi_1, \dots, \Pi_i$ is iteratively defined as follows. Initially, let $D=(\clause)$, $M=\psi$, and $\mathfrak{C}=\{\clause\}$. Throughout, $M$ will serve as a pointer to some matrix and $\mathfrak{C}$ will keep track of all matrices whose origin has not been resolved yet. Let $j = i$ and do the following while $j > 0$.
                \begin{itemize}
            \item If $v_j$ was forgotten via \emph{Rule \ref{itm:R1}}, reduce $j$ by one and proceed.
            \item If $v_j$ was forgotten via \emph{Rule \ref{itm:R2}}, then there is some $\tilde{\psi}$ in some $\pi \in \mathcal{F}^{(j-1)}$ for which $M=\res{\tilde{\psi}}{v_j}$ with the following property: For all $C \in \mathfrak{C}$ with $C \notin \tilde{\psi}$, there exist $C_1^{(C)}, C_2^{(C)} \in \tilde{\psi}$ such that $C$ is the unique clause in $\R{\{C_1^{(C)}, C_2^{(C)}\}}{v_j}$. Intuitively, this means that every clause $C$ whose origin has not been explained yet, can be explained now if $C \notin \tilde{\psi}$, by linking it to the two clauses whose resolvent it is. 
            We extend $D$ to the left by all clauses in $E=\{C_1^{(C)}, C_2^{(C)} \mid C \in \mathfrak{C}, C \notin \tilde{\psi}\}$, set $M=\tilde{\psi}$ and update $\mathfrak{C}=(\mathfrak{C}\setminus \{C \mid  C \notin \tilde{\psi}\}) \cup E$. Finally, we reduce $j$ by one.
                         \item If $v_j$ was forgotten via \emph{Rule \ref{itm:R3}}, then there is some $\tilde{\psi}$ in some $\pi \in \mathcal{F}^{(j-1)}$ such that $M=\red{\tilde{\psi}}{v_j}$ with the following property: For all $C \in \mathfrak{C}$ with $C \notin \tilde{\psi}$, there exists $C^{(C)} \in \tilde{\psi}$ such that $C$ is the unique clause in $\red{\{C^{(C)}\}}{v_j}$. That is, each such clause can be explained now by linking it to a clause from which it originated via reduction. We extend $D$ to the left by all clauses in $E=\{C^{(C)} \mid C \in \mathfrak{C}, C \notin \tilde{\psi}\}$, set $M=\tilde{\psi}$ and update $\mathfrak{C}=(\mathfrak{C}\setminus \{C \mid  C \notin \tilde{\psi}\}) \cup E$.  Finally, we reduce $j$ by one.
                         \item If $v_j$ was forgotten via \emph{Rule \ref{itm:R4}}, then there is some $\tilde{\psi}$ in some $\pi \in \mathcal{F}^{(j-1)}$ and some $\delta\in \langle \var{\Q^{(j-1)}\cap \depless{v_i}}\rangle$ such that for all $C \in \mathfrak{C}$ with $C \notin \tilde{\psi}$, there is some $C^{(C)} \in \tilde{\psi}$ for which $C = C^{(C)}\rvert_{\delta}$. Extend $D$ to the left by all clauses in $E=\{C^{(C)} \mid C \in \mathfrak{C}, C \notin \tilde{\psi}\}$, set $M=\tilde{\psi}$ and update $\mathfrak{C}=(\mathfrak{C}\setminus \{C \mid  C \notin \tilde{\psi}\}) \cup E$.  Finally, we reduce $j$ by one.
                     \end{itemize}
        When $j=0$, extend $D$ to the left by all remaining clauses in $\mathfrak{C}$.
    \end{longdefinition}
    Note that a certifying sequence $D$ is minimal in the sense that each clause $C_j$ with $j \in [t-1]$ is used to derive a subsequent clause in $D$, either by resolution, reduction, or partial assignment. 
    \fi
    We now prove that certifying sequences obey a certain separation property  that we will exploit in the proof of Claim~\ref{cl:equiv}.
    \begin{claim}
    \label{cl:separation}
        Let $i \in [n]$, $p$ be the parent of $\forget{v_i}$, and $\clause$ be a clause that occurs in some matrix $\psi \in \pi$ for some $\pi \in \mathcal{F}^{(i)}$.
        Let $D=(C_1, \dots, C_t)$ be a certifying sequence of $\clause$ that was produced by the partial derivation sequence $\Pi_0, \Pi_1, \dots, \Pi_i$.
                                                                                If $\clause$ contains a variable in $S=\var{\phi} \setminus \mathbf{T}_{\leq p}$, then $\bigcup_{j \in [t]}\var{C_j} \subseteq \var{\phi}\setminus \mathbf{T}_{< p}$, that is, no clause in $D$ contains a variable from $\mathbf{T}_{<p}$.
        \end{claim}
        \iflong
        \begin{claimproof}
            We prove this claim by carefully analyzing the properties of $D$. 
            \begin{description}
                \item[Observation 1] \textbf{No clause in $D$ contains only variables in $\rchi(p)$.} Suppose not and let $C'$ be the last clause in $D$ that contains only variables in $\rchi(p)$. Since $C_t=\clause$ contains a variable from $S$, $C'$ is to the left of $C_t$ in $D$, and since $D$ is minimal, $C'$ has been used to eventually obtain $C_t$. However, note that no clause that occurs behind $C'$ in $D$ has been derived from $C'$ by performing resolution or reduction, since $C'$ contains only variables in $\rchi(p)$, that is, variables that have not yet been forgotten. On the other hand, no clause that occurs behind $C'$ in $D$ has been derived from $C'$ by applying a partial assignment to $C'$ since this clause would be either empty (and therefore not useful to derive $C_t$, contradicting the minimality of $D$) or contain only variables in $\rchi(p)$ (contradicting the fact that $C'$ is the last clause in $D$ with this property).
                \item[Observation 2] \textbf{No clause in $D$ contains a variable in $\mathbf{T}_{<p}$ and a variable in~$S$.} We show the statement via induction. Clearly, all original clauses in $\phi$ satisfy the property, since $\tree$ is a trunk-aligned tree decomposition of $G_{\Phi}$, which means that the node $p$ acts as a separator in $\phi$. Let us now assume that there was a clause in $D$ that contained variables of both sets, $\mathbf{T}_{<p}$ and $S$, and let $C'$ be the first such clause in $D$. Since neither a reduction nor an assignment could have introduced this new neighborhood, there must have been two clauses, $C_{\ell}$ and $C_{m}$, left of $C'$ in $D$ whose resolvent is $C'$ and, without loss of generality, $C_{\ell}$ contained a variable in $\mathbf{T}_{<p}$ and $C_{m}$ contained a variable in $S$. The variable that we resolved over must have been in $\rchi(p)$ as otherwise, $C'$ would not have been the first clause in $D$ to contain a variable from both sets, $\mathbf{T}_{<p}$ and $S$. But this is contradicting the fact that no variable in $\rchi(p)$ has been forgotten yet.
            \end{description}
            Equipped with these observations, we now proceed to prove that no clause in $D$ contains a variable from $\mathbf{T}_{<p}$. Suppose not and let $C'$ be the right-most clause in $D$ that contains a variable $v$ from $\mathbf{T}_{<p}$. Note that $C' \neq C_t$ since, by definition, $C_t$ contains a variable in~$S$ and Observation 2 holds. Hence, by the minimality of $D$, $C'$ is used to derive $C_t$. If~$v$ was dropped from $C'$ via reduction or by assigning it, we would have obtained a clause that is either empty or contains only variables in $\rchi(p)$ (since $C'$ was assumed to be the last clause in $D$ that contains a variable in $\mathbf{T}_{<p}$). The first case contradicts the minimality of~$D$, the latter case contradicts Observation~1. Thus, $C'$ was used in a resolution step and suppose for a moment that the respective pivot variable was not $v$. Then, since $C'$ is the last clause in $D$ that contains $v$, the resolvent must have been tautological and thus got dropped, which contradicts the minimality of $D$. Hence, the pivot variable was $v$, implying that the clause $C''$ we resolved $C'$ with must have contained $v$ in opposite polarity. Note that their respective resolvent does not contain any variable in $S$, since, by Observation 2,  neither did $C'$ nor $C''$.
            Moreover, the resolvent does not contain any variable from~$\mathbf{T}_{<p}$ since we assumed $C'$ to be the right-most clause for which this is the case. Consequently, the resolvent of $C'$ and $C''$ contains only variables in $\rchi(p)$, contradicting Observation 1.
        \end{claimproof}
        \fi
        \ifshort
        \begin{claimproof}[Proof Sketch]
                     We prove this claim by carefully analyzing the properties of $D$. In particular, we can make two key observations: (1) No clause in $D$ contains only variables in $\rchi(p)$, and (2) no clause in $D$ contains both a variable in $\mathbf{T}_{<p}$ and a variable in~$S$.
            
            Equipped with these observations, we now proceed to prove that no clause in $D$ contains a variable from $\mathbf{T}_{<p}$. Suppose not and let $C'$ be the right-most clause in $D$ that contains a variable $v$ from $\mathbf{T}_{<p}$. Note that $C' \neq C_t$ since, by definition, $C_t$ contains a variable in~$S$ and Observation 2 holds. Hence, by the minimality of $D$, $C'$ is used to derive $C_t$. If~$v$ was dropped from $C'$ via reduction or by assigning it, we would have obtained a clause that is either empty or contains only variables in $\rchi(p)$ (since $C'$ was assumed to be the last clause in $D$ that contains a variable in $\mathbf{T}_{<p}$). The first case contradicts the minimality of~$D$, the latter case contradicts Observation~1. Thus, $C'$ was used in a resolution step. 
            
            Suppose for a moment that the respective pivot variable was not $v$. Then, since $C'$ is the last clause in $D$ that contains $v$, the resolvent must have been tautological and thus got dropped, which contradicts the minimality of $D$. Hence, the pivot variable was $v$, implying that the clause $C''$ we resolved $C'$ with must have contained $v$ in opposite polarity. Note that their respective resolvent does not contain any variable in $S$, since, by Observation 2,  neither did $C'$ nor $C''$.
            Moreover, the resolvent does not contain any variable from~$\mathbf{T}_{<p}$ since we assumed $C'$ to be the right-most clause for which this is the case. Consequently, the resolvent of $C'$ and $C''$ contains only variables in $\rchi(p)$, contradicting Observation 1.
        \end{claimproof}
        \fi 
    As a consequence of Claim~\ref{cl:separation}, the existence of a clause $\clause$ that contains some variable in~$S$ does not depend on the assignment that was chosen for the variables in $V \cap \mathbf{T}_{< p}$, since none of these variables were used to derive $\clause$.
    Equipped with this insight, we now proceed to prove the aforementioned sufficient condition for the equivalence of $\psi$ and $\psi'$.
    \begin{claim}\label{cl:equiv}
        If $U_{\psi}=U_{\psi'}$ and $\alpha_1'=\alpha_2'$, then $\psi \equiv \psi'$.
    \end{claim}
    \begin{claimproof}
        We need to show that under the stated assumptions, $\psi$ and $\psi'$ consist of the same set of clauses. Observe that all clauses in $\psi$ and $\psi'$ contain only variables in $\var{\phi} \setminus (\mathbf{T}_{<p} \cup V )$; in particular, each clause either only contains variables in $\var{U}$ or it contains at least one variable from $S=\var{\phi}\setminus \mathbf{T}_{\leq p}$; on the former, $\psi$ and $\psi'$ agree, since $U_{\psi}=U_{\psi'}$. Without loss of generality, let us now consider an arbitrary clause~$\clause$ that occurs in $\psi$ and that contains a variable from $S$. 
                 By Claim~\ref{cl:separation}, no variable in $\mathbf{T}_{<p}$ was used to derive~$\clause$; hence, every partial assignment that was applied to derive $\clause$ was only with respect to variables in $V \setminus \mathbf{T}_{<p}=V \cap \rchi(p)$. Since $\alpha_1'=\alpha_2'$ by assumption, and reduction and resolution steps are applied uniformly across all matrices at each time step $i$, we conclude that the same certifying sequence must have produced $\clause$ in $\psi'$. Consequently, $\psi$ and $\psi'$ also agree on clauses that contain a variable from $S$. Thus, $\psi\equiv \psi'$ which concludes the proof.
    \end{claimproof}
    
We conclude the proof of Lemma~\ref{lem:equiv_formulas} by using Claim~\ref{cl:equiv} to bound the number of matrices in each $\pi \in \mathcal{F}^{(i)}$.
 Since there are at most $2^{|U|}\leq 2^{3^{k+1}}$ distinct choices for $U_{\psi}$ and at most $2^{|V \cap \rchi(p)|} \leq 2^{k+1}$ distinct choices for an assignment $\alpha' \in \langle V \cap \rchi(p) \rangle$, we conclude from Claim~\ref{cl:equiv} that $|\pi| \leq 2^{3^{k+1}}\cdot 2^{k+1} \in 2^{2^{\oh{k}}}$ for any $\pi \in \mathcal{F}^{(i)}$. Consequently, there are at most $2^{2^{2^{\oh{k}}}}$ possible choices for $\pi$, that is,  $|\mathcal{F}^{(i)}| \leq 2^{2^{2^{\oh{k}}}}$.
\end{proof}

We are now ready to show that the Rules~\ref{itm:R1} to~\ref{itm:R4} can be applied efficiently. 
 \begin{lemma}\label{lem:rules_efficient}
There is a computable function $f$ with the following property. Given a partial derivation sequence $\Pi'=\Pi_0, \dots, \Pi_{i-1}$ of $\qbf$ and a trunk-aligned tree decomposition $\tree=(\mathbf{T}, \rchi)$ of $(\Phi, \D)$ with width $k=\btw{\Phi, \D}$, we can derive $\Pi_i$ in time $|\mathcal{T}|+f(k)\cdot |\Phi|$.
 \end{lemma}
\begin{proof}     We perform a case analysis based on the type of $v_i \in\V$ that is forgotten to obtain~$\Pi_i$. Note that the relevant case can be identified in time $|\mathcal{T}|$.
    \begin{description}
        \item[Case 1] If $v_i \notin \var{\Q^{(i-1)}}$, then Rule~\ref{itm:R1} sets $\Pi_i=\Pi_{i-1}$ in constant time.
        \item[Case 2] If $v_i \in \vare{\Q^{(i-1)}}$ with $\{v \in \var{\Q^{(i-1)}} \mid v_i \dep v\} \cap \rchi(\forget{v_i})=\emptyset$, then we apply Rule~\ref{itm:R2}, that is, we resolve over $v_i$ in all matrices $\psi$ that are contained in some $\pi \in \mathcal{F}^{(i-1)}$.  
        By Lemma~\ref{lem:bounded_neighborhood}, $v_i$ shares a clause with at most $k$ other variables in each such $\psi$. Since these variables can occur either positively, negatively, or not at all, both $\psi_{v_i}$ and $\psi_{\overline{v_i}}$ contain at most $3^k$ clauses and each of these clauses is of size at most $k+1$. Consequently, the resolvent $\R{\psi}{v_i}$ replaces the clauses in $\psi_{v_i}\cup\psi_{\overline{v_i}}$ by at most $3^k$ clauses 
        (the set of all clauses that the $k$ neighbors of $v_i$ can possibly form), and computing $\R{\psi}{v_i}$ takes time $\oh{{(3^{k})}^2k}$. Moreover, by Lemma~\ref{lem:equiv_formulas}, we can afford to compute $\R{\psi}{v_i}$ for each matrix $\psi$ that is contained in some $\pi \in \mathcal{F}^{(i-1)}$.
        \item[Case 3] If $v_i \in \vara{\Q^{(i-1)}}$ with $\{v \in \var{\Q^{(i-1)}} \mid v_i \dep v\} \cap \rchi(\forget{v_i})=\emptyset$, then, applying Rule~\ref{itm:R3}, we reduce $v_i$ from all matrices $\psi$ that are contained in some $\pi \in \mathcal{F}^{(i-1)}$. 
 This may be done by a linear pass over all matrices in all sets of $\mathcal{F}^{(i-1)}$. At the same time, we remark that by Lemma~\ref{lem:bounded_neighborhood}, $v_i$ shares a clause with at most $k$ other variables in each such~$\psi$. Hence, in each $\psi$, the variable~$v_i$ will only be reduced from at most $2\cdot 3^k$ clauses.        Thus, computing $\red{\psi}{v_i}$ can be done in time $\oh{3^k}$ under the assumption that variable-clause adjacencies can be looked up in constant time, and by Lemma~\ref{lem:equiv_formulas} we can afford to do so for each matrix $\psi$ that is contained in some $\pi \in \mathcal{F}^{(i-1)}$.
        \item[Case 4] If $v_i \in \var{\Q^{(i-1)}}$ with $S\neq \emptyset$ for $S=\{v \in \var{\Q^{(i-1)}} \mid v_i \dep v\}\cap \rchi(\forget{v_i})$, then we apply Rule~\ref{itm:R4}. Since $S$ is not empty, there is at least one $v \in \var{\Q^{(i-1)}}$ with $v_i \in \depless{v}$ such that $v \in \rchi(\forget{v_i})$, that is, $v_i$ violates Property~\ref{itm:P1}. Hence, by Property~\ref{itm:P2}, $\forget{v_i}\in \trunk$ and for all $v' \in \depless{v_i}$, we have $v' \in \mathbf{T}_{\leq \forget{v_i}}$. 
        Observe that each $v' \in( \var{\Q^{(i-1)}} \cap \depless{v_i})$ has not yet been forgotten, that is, $v' \in \rchi(\forget{v_i})$. From the definition of strategy extension in Section~\ref{sec:rules}, it follows that computing $\strext{\pi}{v_i}$ for a single $\pi \in \mathcal{F}^{(i-1)}$ takes time $(|\A|)^{|\pi|}\cdot|\B|\cdot {|\Phi|}$. Since $|\var{\Q^{(i-1)}} \cap \depless{v_i}|\leq|\rchi(\forget{v_i})|\leq k+1$, we have $|\B|\leq 2^k$ and $|\A|\leq 2^{2^{\oh{k}}}$ by the definition of $\A$ and $\B$ in strategy extension. Finally, recall that Lemma~\ref{lem:equiv_formulas} bounds $|\pi|\leq 2^{2^{\oh{k}}}$ and $|\mathcal{F}^{(i-1)}|\leq 2^{2^{2^{\oh{k}}}}$. Consequently, we can compute $\strext{\pi}{v_i}$ in time $2^{2^{2^{\oh{k}}}}\cdot |\Phi|$ and can afford to do so for each $\pi \in \mathcal{F}^{(i-1)}$.
                           \qedhere
    \end{description}
\unskip\end{proof}
With this final piece in hand, we are now ready to establish Theorem~\ref{thm:fpt_btw}.

\begin{proof}[Proof of Theorem~\ref{thm:fpt_btw}]
First, we compute an elimination ordering~$\V=(v_1, \dots, v_n)$ that agrees with $\tree$ and store $\rchi(\forget{v_i})$ for each $i \in [n]$. This can be done in time $\oh{|\tree|}$. Next, we initialize $\Pi_0=(\Q^{(0)}, \mathcal{F}^{(0)})$ with $\Q^{(0)}=\Q$ and $\mathcal{F}^{(0)}=\{\{\phi\}\}$. Now, for each $i \in [n]$, we compute $\Pi_i$ by using one of the Rules~\ref{itm:R1} to \ref{itm:R4} depending on the variable type of~$v_i \in \V$. By Lemma~\ref{lem:rules_efficient}, this sequence of sets up to $\Pi_n$ can be computed in time $f(k)\cdot |\Phi|\cdot |\mathcal{T}|$.
 It follows from Lemma~\ref{lem:main_invariant} that $\Phi$ is true if and only if there exists a set $\pi \in \mathcal{F}^{(n)}$ of matrices such that each matrix $\psi \in \pi$ is true. 
\iflong
Moreover, by Observation~\ref{obs:total_elim}, each $\pi \in \mathcal{F}^{(n)}$ consists only of matrices~$\psi$ that contain the empty clause $\bot$ or are themselves empty.
\fi
\ifshort
Moreover, since~\ref{itm:R1} to~\ref{itm:R4} eliminate variables exhaustively, each matrix in $\pi \in \mathcal{F}^{(n)}$ is either empty or contains the  empty clause $\bot$.
 \fi
By Lemma~\ref{lem:equiv_formulas}, the number of QBFs and sets of QBFs in $\mathcal{F}^{(n)}$ is bounded by a function of $k$; hence, we can efficiently detect whether
there exists a set $\pi \in \mathcal{F}^{(n)}$ such that no matrix $\psi \in \pi$ contains the empty clause $\bot$.
  If so, we return TRUE, otherwise we return FALSE.
\end{proof}

\section{Concluding Remarks}
Our work closes the previous rift between prefix pathwidth and dependency treewidth, two incomparable measures for the tree-likeness of QBFs. Bilateral treewidth provides not only a unified explanation for why prefix pathwidth and dependency treewidth enable fixed-parameter QBF evaluation, but genuinely extends them by covering instances that are beyond the reach of any of the previously developed width parameters for QBF. Moreover, our evaluation algorithm produces fixed-parameter sized proofs as a witness and thus preserves a key advantage of dependency treewidth. In future work, it would be interesting to investigate whether the asymptotic running time of the presented QBF evaluation algorithm can be improved, or whether it is tight under the Exponential Time Hypothesis.

From the viewpoint of proof complexity, it may be of interest to formally analyze these proofs and, if possible, link them to existing proof systems. 
From an empirical perspective, it may be interesting to see whether any of the insights employed in our proofs can be adapted to speed up the solution of instances which are expected to be ``tree-like''. We remark that---similarly as for the preceding work on prefix pathwidth~\cite{EGO20}---the theoretical upper bound in the function $f$ is triple-exponential in the bilateral treewidth (cf.\ Lemma~\ref{lem:equiv_formulas}), and hence is not of practical significance.

Moreover, a major question left for future work is whether the decompositions underlying bilateral treewidth can be computed in fixed-parameter time. None of the known techniques for computing or approximating treewidth~\cite{A87,R92,B93,BDDFLP16,K21} seem directly applicable to bilateral treewidth, and in fact the question remains open even for the previously studied notions of prefix pathwidth and dependency treewidth. Nevertheless, as illustrated in the proof of Theorem~\ref{thm:parity}, there are natural classes of QBFs which remained out-of-reach for previous width measures but where constant-width decompositions can be computed easily. A natural stepping stone towards resolving this question would be to first obtain an \XP~algorithm for computing trunk-aligned tree decompositions of minimum width.

\bibliography{references}

\end{document}